\documentclass[12pt]{article}

\usepackage{eurosym}
\usepackage{graphicx}
\usepackage{enumerate}
\usepackage{amsmath}
\usepackage{amsfonts}
\usepackage{amssymb}
\usepackage{amsfonts}
\usepackage{amssymb,amsmath,amsthm,epsfig,graphicx,natbib,subfigure,hyperref}
\usepackage{geometry}
\usepackage{caption}
\usepackage{color}
\usepackage{enumerate}
\usepackage{setspace}
\usepackage[textwidth=30mm]{todonotes}
\usepackage{sgame, tikz}
\usepackage{multirow,array}
\usepackage{rotating}
\usepackage{graphicx}
\usepackage{subfigure}
\usepackage{float}
\usepackage{subcaption}
\usepackage{amsthm}

\definecolor{darkblue}{rgb}{0.0, 0.0, 0.55}
\hypersetup{colorlinks,linkcolor={darkblue},citecolor={darkblue},urlcolor={darkblue}} 
\allowdisplaybreaks

\bibpunct[: ]{(}{)}{;}{a}{}{,}

\usetikzlibrary{arrows.meta}
\tikzset{tail/.tip={Straight Barb[reversed, length=1.5pt]}}

\newcommand{\xright}[2][]{\tikz[baseline,anchor=base]{\node[inner sep=0pt,#1](a){$#2$};
\draw[tail-]([shift={(1pt,2pt)}]a.north west)--([shift={(-1pt,2pt)}]a.north east);}}

\newcommand{\xleft}[2][]{\tikz[baseline,anchor=base]{\node[inner sep=0pt](a){$#2$};
\draw[-tail]([shift={(1pt,2pt)}]a.north west)--([shift={(-1pt,2pt)}]a.north east);}}

\newcommand{\xboth}[2][]{\tikz[baseline,anchor=base]{\node[inner sep=0pt](a){$#2$};
\draw[tail-tail]([shift={(1pt,2pt)}]a.north west)--([shift={(-1pt,2pt)}]a.north east);}}

\newtheorem{thm}{Theorem}
\newtheorem{pro}{Proposition}
\newtheorem{lem}{Lemma}
\newtheorem{alg}{Algorithm}
\newtheorem{ass}{Assumption}
\newtheorem{corr}{Corollary}
\newtheorem{ex}{Example}
\newtheorem{defin}{Definition}
\newtheorem{ob}{Observation}

\newenvironment{proposition}{\begin{pro}}{\end{pro}}
\newenvironment{lemma}{\begin{lem}}{\end{lem}}

\newenvironment{definition}{\begin{defin}}{\end{defin}}

\newcommand{\be}{\begin{equation}}
\newcommand{\ee}{\end{equation}}

\hypersetup{colorlinks=true,linkcolor=blue,urlcolor=blue,citecolor=blue}
\usepackage{soul}

\begin{document}

\title{\sc	Influence and Connectivity in  Networks: A Generating Function Approach}
\author{Yang Sun\thanks{
CCBEF, Southwestern University of Finance and Economics, China. \textit{Email}: sunyang789987@gmail.com} \and Wei Zhao\thanks{School of Economics and Management, Tsinghua University, China. \textit{\ Email}: wei.zhao@outlook.fr } \and Junjie Zhou\thanks{
School of Economics and Management, Tsinghua University, China. \textit{Email}: zhoujj03001@gmail.com} }
\date{\today  }
\maketitle

\begin{abstract}


Many widely used network centralities are based on counting walks that meet specific criteria. This paper introduces a systematic framework for walk enumeration using generating functions. We introduce a first-passage decomposition that uniquely divides any walk passing through specified nodes or links into two components: a first-reaching walk and a subsequent walk. This decomposition yields a system of interconnected equations that relate three  disjoint categories of walks: unrestricted walks, walks that avoid specific elements, and walks that pass through designated sets. The framework offers a range of applications, including evaluating the effects of structural interventions—such as node or link modifications—on network walks, generalizing target centrality to multi-receiver scenarios in information networks, and comparing different strategies for adding links.

\end{abstract}

\newpage 

\section{Introduction}

Since \cite{Katz1953}, walk-based statistics have been adopted to describe nodes' positions within a network and network structure. As \cite{borgatti2005} emphasizes, walk-based statistics are suitable for measuring information flow and interpersonal influence within networks. Given the micro-foundation of Katz-Bonacich centrality provided by \cite{Ballester2006}, walk-based statistics are prevalent in characterizing strategic interaction and equilibrium behavior in network games. Many well-known walk-based statistics are based on counting a specific class of walks. The past literature derives these statistics case by case. This paper provides a rigorous and unified framework for walk-counting based on generating functions. 

Consider an enumeration question which counts the number of ways to add up two elements in two combinatorial class respectively, achieving a given size. To tackle this question, the generating function encodes size information of combinatorial class through formal power series. Two properties of this generating function are listed. The generating function of disjoint union is the sum while that of cartesian is the product, of the generating function of each combinatorial class. 

Given a network, a walk is a vector which encodes the nodes it traverses and the length of a walk counts the number of nodes it traverses. The concatenation operation is to merge two walks, where one walk's end meets the other's start. This operation is extended between two walk sets, which merges individual walks in each set. The concatenation (between two walk sets) can be uniquely decomposed if each merged walk is the concatenation of a unique pair of walks. Taking walk sets as enumeration classes with walks' length as size, the concatenation operation is equivalent to cartesian if and only if unique decomposition is achieved. 

To achieve unique decomposition, we extend the first passage decomposition prevalent in stochastic process. Fix a subset of nodes or links, a natural partition of a walk set is based on whether or not passing this subset of nodes or links. Additionally, the set of walks which must pass the subset of nodes (or links) is a concatenation with unique decomposition of two components (1) the set of walks which first reaches one of the nodes (first passes one of the links), and (2) the set of walks which start from the reached nodes. Accordingly, a linear relationship is established for the generating functions among these three components (including the original walk set). The generating function of one component can therefore represent the other two. Note that, whenever convergence is imposed, the generating function of a walk set is exactly the aggregate of walks in the set, geometrically discounted by their length. 

Several applications of these linear relationships are examined. First, fix a node, if taking the set of its neighbours or the set of links with its neighbours for the partition, we can therefore link the aggregate walks between a pair of nodes and the inverse of the Leontief matrix. Second, by taking an arbitrary single node for the partition, we can derive the aggregate of walks which does pass this node and therefore its inter-centrality (see \citealt{Ballester2006}). A direct extension is to take subset of nodes for partition, which then derives the group inter-centrality, proposed by \cite{Ballester2010}. Third, take an arbitrary link for the partition. On one hand, through taking the set of original walks as given, we can characterize the effect of link removal on aggregate walks and therefore identify the key link (see \citealt{Ballester2010}). On the other hand, through taking the set of walks which avoids the link as given, we can characterize the effect of link/bridge build up and therefore identify the optimal way (see \citealt{Konig2014}, \citealt{sun2023}). Further, we can explicitly derive the overall impact of general structural intervention, including adding and deleting multiple links at same time. Another application is drawn on information transmission. The diffusion centrality counts the expected reception of a given node (See \citealt{banerjee2013diffusion}). \cite{Bramoulle2018} additionally imposes intuitive non-retransmission assumptions, including the non-retransmission for the source, the target and both, and propose the target centrality. We first extend to group target centrality. Then we propose the intermediary centrality, to characterize the impact of a node's removal on information transmission. Finally, we compare link construction between two nodes $i$ and $j$ where their neighborhoods exhibit a nesting relationship ($N_i \supseteq N_j$). By uniquely decomposing the set of all walks that traverse newly constructed links exactly $n+1$ times into concatenations of walks passing through these links exactly once and walks passing through them exactly $n$ times, we establish an iterative proof. This approach demonstrates that forming connections between any external node set and node $i$ generates strictly more walks of any arbitrary length throughout the network than establishing equivalent connections to node $j$.

Our methodological innovation lies in applying generating functions to the decomposition of walks. The first passage decomposition, initially developed in stochastic processes to analyze threshold-crossing events and related phenomena (see \citealp{Norris2004}), is adapted here for unweighted networks. Generating functions have proven powerful across diverse fields including quantum field theory (\citealp{ZinnJustin2002}), statistical mechanics (\citealp{THOMPSON1972}), and biological population dynamics (\citealp{Kimmel2015}). In graph theory, walk generating functions have traditionally been employed to analyze how structural interventions affect the determinant of adjacency matrices (\citealp{Godsil1992}; \citealp{Rowlinson1996}). We advance this literature by generalizing the walk generating function from the set of all $i$-$j$ walks to arbitrary walk sets, and combining this with first passage decomposition to systematically analyze how specific nodes and links influence walks between node pairs.

Our work connects to walk-based centrality measures, which have long been fundamental tools for understanding network structure and influence. The eigenvector centrality (\citealp{Bonacich1972}) and its variant Katz-Bonacich centrality (\citealp{Katz1953}; \citealp{bonacich1987}) are among the earliest and most influential measures based on counting walks. Recently, \cite{Bloch2023} used node walk statistics to axiomatize various centrality measures, including eigenvector and Katz-Bonacich centrality. These measures have found wide applications across economic contexts, including production networks (\citealp{Acemoglu2012}; \citealp{baqaee2018cascading}; \citealp{liu2019industrial}), pricing of social products (\citealp{candogan2012optimal}; \citealp{Fainmesser2016}; \citealp{BLOCH2013243}), and multi-agent contracting (\citealp{Mayol2024}). A seminal contribution in this area is \cite{Ballester2006}, which provided a microfoundation for Katz-Bonacich centrality by demonstrating that players' equilibrium behavior in linear quadratic network games coincides with their Katz-Bonacich centrality. Subsequent research has extended this framework to multi-dimensional actions (\citealp{D2x}), games with congestion effects (\citealp{CURRARINI201740}), and games where payoffs depend on local-average efforts (\citealp{Ushchev2020}). The equilibria of these modified models are all characterized by variants of walk-based centralities. Recent developments include information diffusion centrality (\citealp{banerjee2013diffusion}; \citealp{Cruz2017politician}), community centrality (\citealp{Benzi2013}), and targeting centrality(\citealp{Bramoulle2018}), further demonstrating the fundamental importance of walk-counting in network analysis.

Building on these foundations, a related strand of literature examines how structural interventions affect centrality measures, providing insights for comparative statics and policy design. This includes analyses of key players (\citealp{Ballester2006}), key links (\citealp{Ballester2010}; \citealp{Zenou2014}), key groups (\citealp{Ballester2010}), key leaders (\citealp{Zhou2015}), and key bridges (\citealp{Golub2010}; \citealp{sun2023}). While these studies primarily employ matrix analysis to examine individual network elements' effects on Katz-Bonacich centrality, our paper provides a unified approach for analyzing arbitrary sets of nodes and links, thereby generalizing these classical results. Notably, \cite{Bramoulle2018} introduced target centrality to analyze information transmission between specific node pairs under non-retransmission constraints—constraints that can be viewed as structural interventions in information diffusion networks. We extend this framework to multiple information receivers and systematically identify key mediators in information diffusion processes.

The remainder of this paper is organized as follows. Section 2 introduces the preliminaries of generating functions. Section 3 presents our walk concatenation framework and main result. Section 4 demonstrates key applications of our methodology. Section 5 concludes the paper. All proofs are relegated to the Appendix.

\section{Preliminaries: Generating Functions}
\label{sec:preliminary}

Generating functions provide a powerful framework for counting discrete objects that can be finitely (or countably) described by construction rules. We begin by establishing the fundamental concepts.

A combinatorial class $\mathcal{A}$ is a denumerable set equipped with a size function $f_{\mathcal{A}}:\mathcal{A}\rightarrow \mathbb{Z}_{0}^{+}$ such that the inverse image of any integer is finite, i.e., $|f_{\mathcal{A}}^{-1}(t)|<\infty$ for any $t \geq 0$. For any non-negative integer $t$, $|f_{\mathcal{A}}^{-1}(t)|$ counts the number of objects in $\mathcal{A}$ that have size $t$. The notion of “size” represents a numerical measure assigned to each object, which varies depending on context—for combinations, it might be the number of chosen elements; for words, it might be the length.

Consider two classes $\mathcal{B}$ and $\mathcal{C}$. The \emph{Cartesian product} forms ordered pairs
\begin{equation*}
\mathcal{A}\equiv\mathcal{B}\times \mathcal{C}=\{\alpha =(\beta,\gamma) \mid \beta \in \mathcal{B}, \gamma \in \mathcal{C}\} 
\end{equation*}
with the size of a pair $\alpha =(\beta,\gamma)$ being defined by
\begin{equation*}
f_{\mathcal{A}}(\alpha) = f_{\mathcal{B}}(\beta) + f_{\mathcal{C}}(\gamma).
\end{equation*}
As a result, the number of objects in $\mathcal{A}=\mathcal{B}\times \mathcal{C}$ that have size $t$ is
\begin{equation}
|f_{\mathcal{A}}^{-1}(t)| = \sum_{k=0}^{t} |f_{\mathcal{B}}^{-1}(k)| \cdot |f_{\mathcal{C}}^{-1}(t-k)|.  \label{eq:size*}
\end{equation}

When $\mathcal{B}\cap \mathcal{C}=\emptyset$, the \emph{union} forms a combinatorial class $\mathcal{A}\equiv\mathcal{B}\cup \mathcal{C}$ with size defined by
\begin{equation*}
f_{\mathcal{A}}(\omega) = 
\begin{cases}
f_{\mathcal{B}}(\omega), & \text{if }\omega \in \mathcal{B} \\
f_{\mathcal{C}}(\omega), & \text{if }\omega \in \mathcal{C}
\end{cases}
\end{equation*}
As a result, the number of objects in $\mathcal{A}=\mathcal{B}\cup \mathcal{C}$ that have size $t$ is
\begin{equation}
|f_{\mathcal{A}}^{-1}(t)| = |f_{\mathcal{B}}^{-1}(t)| + |f_{\mathcal{C}}^{-1}(t)|.  \label{eq:size+}
\end{equation}

A generating function serves as a clothesline on which we display the counts of objects arranged by their sizes.
\begin{definition}
The generating function of a combinatorial class $\mathcal{A}$ is defined as: 
\begin{equation*}
G[[\mathcal{A};x]] = \sum_{t\geq 0} |f_{\mathcal{A}}^{-1}(t)| \cdot x^{t},
\end{equation*}
where $x$ is an indeterminate.
\end{definition}

In this formulation, the coefficient of $x^{t}$ directly corresponds to the number of objects in $\mathcal{A}$ having size $t$. We treat generating functions algebraically as formal power series, with $x$ serving as a symbolic placeholder, independent of convergence considerations.

The counting principles for size functions (Equations \eqref{eq:size*} and \eqref{eq:size+}) under the product and union constructions translate directly into corresponding properties for generating functions:

For combinatorial classes $\mathcal{B}$ and $\mathcal{C}$:
\begin{enumerate}[(i)]
\item \textbf{Product Property:} $G[[\mathcal{B}\times \mathcal{C};x]] = G[[\mathcal{B};x]] \cdot G[[\mathcal{C};x]]$;
\item \textbf{Sum Property:} $G[[\mathcal{B}\cup \mathcal{C};x]] = G[[\mathcal{B};x]] + G[[\mathcal{C};x]]$ provided $\mathcal{B}\cap \mathcal{C} = \emptyset$.
\end{enumerate}

These properties demonstrate that combinatorial operations correspond directly to algebraic operations on generating functions. The product of generating functions mirrors the combination of objects through Cartesian products, while the addition of generating functions reflects the union of disjoint sets. This correspondence transforms combinatorial problems into algebraic ones, making generating functions an essential tool in our subsequent analysis.
\section{Walks in Networks}
\label{sec:setup}

\subsection{Basic Framework}

Let $N=\{1,2,\ldots,n\}$ denote the set of nodes. We consider unweighted and directed networks, represented by an adjacency matrix $\mathbf{G}=(g_{ij})_{n\times n}$, where $g_{ij}\in \{0,1\}$ indicates the existence of a direct link from node $i$ to node $j$ for any $i,j\in N$. We assume no self-links, i.e., $g_{ii}=0$ for all $i\in N$. Let $\mathcal{G}=\{(i,j):g_{ij}=1\}$ denote the set of links.

A walk $w$ in network $\mathbf{G}$ is a sequence of nodes $w=(i_0,i_1,\ldots,i_K)$ such that $(i_k,i_{k+1})\in \mathcal{G}$ for each $k=0,\ldots,K-1$. The length of a walk, denoted by $\#(w)=K$, equals the number of links in it. For any node $i$, a trivial walk $(i)$ is a walk of length 0 starting from $i$.\footnote{Note that there are multiple length-zero (trivial) walks in the network, one for each node.} Trivial walks are distinct from the empty walk, denoted by $\emptyset$, whose length is undefined. Let $\mathcal{W}$ denote the set of all walks.

For any $i,j\in N$, we define the set of $i$-$j$ walks as 
\begin{equation*}
\mathcal{W}_{ij}:=\{w\in \mathcal{W}:w=(i_{0},i_{1},\ldots ,i_{K})\text{ where }i_{0}=i\text{ and }i_{K}=j\}.
\end{equation*}
Note that $\mathcal{W}_{ij}\cap \mathcal{W}_{kl}=\emptyset$ unless $i=k$ and $j=l$. Moreover, the length-zero walk $(i) \in \mathcal{W}_{ij}$ if and only if $j=i$. For notational convenience, we sometimes use $(i)$ to denote the singleton set containing the length-zero walk.

\begin{definition}[Walk Sets]
\label{def:walks} 
For any subset of nodes $A\subseteq N$ and nodes $i,j\in N$, we define: 
\begin{align*}
\mathcal{W}_{ij}(A) &:= \{w\in \mathcal{W}_{ij}: w=(i_0,i_1,\ldots,i_K) \text{ where } i_k\in A \text{ for some } k=1,\ldots,K-1\}; \\
\mathcal{W}_{ij}(\lnot A) &:= \{w\in \mathcal{W}_{ij}: w=(i_0,i_1,\ldots,i_K) \text{ where } i_k\notin A \text{ for all } k=1,\ldots,K-1\}.
\end{align*}

Moreover, for any subset of links $\mathcal{L}\subseteq \mathcal{G}$, we define: 
\begin{align*}
\mathcal{W}_{ij}(\mathcal{L}) &:=\{w\in \mathcal{W}_{ij}:w=(i_{0},i_{1},\ldots,i_{K})\text{ where }(i_{k},i_{k+1})\in \mathcal{L}\text{ for some }k=0,\ldots,K-1\}; \\
\mathcal{W}_{ij}(\lnot \mathcal{L}) &:=\{w\in \mathcal{W}_{ij}:w=(i_{0},i_{1},\ldots,i_{K})\text{ where }(i_{k},i_{k+1})\notin \mathcal{L}\text{ for all }k=0,\ldots,K-1\}.
\end{align*}
\end{definition}

The set $\mathcal{W}_{ij}(A)$ contains all $i$-$j$ walks that pass through at least one node in set $A$, excluding the starting node $i_{0}=i$ and terminating node $i_{K}=j$. In contrast, $\mathcal{W}_{ij}(\lnot A)$ represents walks that avoid all nodes in $A$, except possibly the starting node $i$ and terminating node $j$. Here, $A$ acts as an absorbing set that walks can enter but not leave. By construction, sets $\mathcal{W}_{ij}(A)$ and $\mathcal{W}_{ij}(\lnot A)$ are complementary for any $A\subseteq N$: 
\begin{equation}
\mathcal{W}_{ij}(A)=\mathcal{W}_{ij}\setminus \mathcal{W}_{ij}(\lnot A),\quad \forall i,j\in N,A\subseteq N.  \label{eq:partition1}
\end{equation}

The set $\mathcal{W}_{ij}(\lnot A)$ corresponds to the widely-studied first passage problem in stochastic processes (see recent survey by \citealp{Redner2023}). Consider when matrix $\mathbf{G}=(g_{ij})_{n\times n}$ is row-normalized to represent transition probabilities in a stochastic process, where $g_{ij}\geq 0$ denotes the probability of transition from state $i$ to $j$. In this context, $\mathcal{W}_{ij}(\lnot A)$ captures all possible walks from $i$ to $j$ that terminate immediately upon reaching any state in set $A$. This concept is fundamental in analyzing stopping times—for instance, in financial markets where trading strategies execute when a stock price first hits a threshold value. The threshold price levels form the absorbing set $A$, and $\mathcal{W}_{ij}(\lnot A)$ represents all possible price trajectories from level $i$ to $j$ that haven't yet triggered the threshold condition. Following the same spirit of first passage analysis, we apply this idea to unweighted networks, studying walks that terminate upon encountering nodes in set $A$.

Length-zero walks require special attention. For any node $j$ and set $A$, we have $(i) \notin \mathcal{W}_{ij}(A)$, while $(i) \in \mathcal{W}_{ij}(\lnot A)$ if and only if $i=j$. The behavior of the length-zero walk $(i)$ illustrates the subtle distinction between $\mathcal{W}_{ii}(\lnot N)$ and $\mathcal{W}_{ii}(\emptyset)$. Specifically, $(i) \in \mathcal{W}_{ii}(\lnot N)$ since $(i)$ doesn't traverse any network nodes, whereas $(i) \notin \mathcal{W}_{ii}(\emptyset)$ by our definition. For the case where $A=N_i$ with $N_i=\{j:g_{ij}=1\}$ being $i$'s out-neighbors, we have 
\begin{equation}
\mathcal{W}_{ij}(\lnot N_i) = 
\begin{cases}
\{(i)\}, & \text{when } j=i \\ 
\{(i,j)\}, & \text{when } j\in N_i \\ 
\emptyset, & \text{otherwise}
\end{cases}
; \quad \mathcal{W}_{ij}(N_i) = 
\begin{cases}
\mathcal{W}_{ij}\setminus \{(i)\}, & \text{when } j=i \\ 
\mathcal{W}_{ij}\setminus \{(i,j)\}, & \text{when } j\in N_i \\ 
\mathcal{W}_{ij}, & \text{otherwise}
\end{cases}
.  \label{eq:neighbors}
\end{equation}

For a set of links $\mathcal{L}\subseteq\mathcal{G}$, $\mathcal{W}_{ij}(\mathcal{L})$ comprises walks from $i$ to $j$ that traverse at least one link in $\mathcal{L}$, while $\mathcal{W}_{ij}(\lnot \mathcal{L})$ consists of walks that avoid all links in $\mathcal{L}$. Since length-zero walks utilize no links, $(i)\notin \mathcal{W}_{ij}(\mathcal{L})$ for any $j$ and $\mathcal{L}$, while $(i)\in \mathcal{W}_{ij}(\lnot \mathcal{L})$ if and only if $i=j$. Analogous to the node case, the behavior of the length-zero walk $(i)$ illustrates the subtle distinction between $\mathcal{W}_{ii}(\lnot \mathcal{G})$ and $\mathcal{W}_{ii}(\emptyset)$: specifically, $(i)\in \mathcal{W}_{ii}(\lnot \mathcal{G})$ since $(i)$ employs no links, whereas $(i)\notin \mathcal{W}_{ii}(\emptyset)$ by definition. As with node-based walk sets, these link-defined walk sets are complementary: 
\begin{equation}
\mathcal{W}_{ij}(\mathcal{L})=\mathcal{W}_{ij}\setminus \mathcal{W}_{ij}(\lnot \mathcal{L}), \quad \forall i,j\in N,\mathcal{L}\subseteq \mathcal{G}.  \label{eq:partition2}
\end{equation}
For the case where $\mathcal{L}_{i}=\left\{ \left( i,k\right) :k\in N_{i}\right\}$ is the set of $i$'s out-links, we have 
\begin{equation}
\mathcal{W}_{ij}(\lnot \mathcal{L}_{i})=%
\begin{cases}
\{(i)\}, & \text{if }j=i \\ 
\emptyset , & \text{otherwise}%
\end{cases}%
;\quad \mathcal{W}_{ij}(\mathcal{L}_{i})=%
\begin{cases}
\mathcal{W}_{ij}\setminus \{(i)\}, & \text{if }j=i \\ 
\mathcal{W}_{ij}, & \text{otherwise}%
\end{cases}%
.  \label{eq:links}
\end{equation}%
This reflects that the only walk not traversing any out-link from its initial node is the length-zero walk at that node.

\subsection{First Passage Decomposition}

We define a non-commutative product operation on walks called concatenation.

\begin{definition} ~
\begin{enumerate}
\item[(i)] For two walks $w=(i_{0},i_{1},\ldots,i_{K})$ and $w^{\prime}=(j_{0},j_{1},\ldots,j_{T})$, their concatenation is defined as 
\begin{equation*}
w\odot w^{\prime}:= 
\begin{cases}
(i_{0},i_{1},\ldots,i_{K},j_{1},\ldots,j_{T}), & \text{if }i_{K}=j_{0}; \\ 
\emptyset, & \text{otherwise.}
\end{cases}
\end{equation*}

\item[(ii)] For walk sets $\mathcal{W}_{1},\mathcal{W}_{2}\subseteq \mathcal{W}$, their concatenation is defined as: 
\begin{equation*}
\mathcal{W}_{1}\odot \mathcal{W}_{2}:=\{w:w=w_{1}\odot w_{2}\text{ for some }w_{1}\in \mathcal{W}_{1}\text{ and }w_{2}\in \mathcal{W}_{2}\}.
\end{equation*}

\item[(iii)] Suppose $\mathcal{W}^{\prime}=\mathcal{W}_{1}\odot \mathcal{W}_{2}$. Then, $\mathcal{W}^{\prime}$ is said to be uniquely decomposed as the concatenation of $\mathcal{W}_{1}$ and $\mathcal{W}_{2}$ if $w_{1}\odot w_{2}=w_{1}^{\prime}\odot w_{2}^{\prime}$ implies $w_{1}=w_{1}^{\prime}$, $w_{2}=w_{2}^{\prime}$, and $w_{1}\odot w_{2}\neq \emptyset$ for all non-empty walks $w_{1},w_{1}^{\prime}\in \mathcal{W}_{1}$, $w_{2},w_{2}^{\prime}\in \mathcal{W}_{2}$.
\end{enumerate}
\end{definition}

When two walks are concatenated with matching endpoints ($i_{K}=j_{0}$), their lengths add up, yielding $\#(w\odot w^{\prime})=\#(w)+\#(w^{\prime})$. The empty walk $\emptyset$ acts as an absorbing element: $w\odot \emptyset=\emptyset\odot w=\emptyset$ for all walks $w\in \mathcal{W}$. Length-zero walks serve as identity elements: $w\odot (i_{K})=(i_{0})\odot w=w$ for any walk $w$.

The concatenation of two sets of walks is defined as the set of all possible concatenations.\footnote{The set of nodes $N$ and walks $\mathcal{W}$, equipped with concatenation, form a small category where objects are nodes $N$, morphisms are walks $\mathcal{W}$, and morphism composition is concatenation, satisfying identity and associativity properties.} This definition extends naturally to concatenating multiple sets of walks. When concatenating walk sets, a crucial distinction emerges between unique and non-unique decompositions. Formally, unique decomposition fails when a walk can be expressed as the concatenation of two different pairs of walks—that is, when $w_{1}\odot w_{2} = w_{1}^{\prime}\odot w_{2}^{\prime}$ where $w_{1} \neq w_{1}^{\prime}$ and $w_{2} \neq w_{2}^{\prime}$. 

The concatenation $\mathcal{W}_{ij}\odot \mathcal{W}_{jj}$ provides a clear example of non-unique decomposition. While this equals $\mathcal{W}_{ij}$, walks in $\mathcal{W}_{ij}$ that pass through node $j$ multiple times can be formed in multiple ways. Consider a simple hub-spoke network where node 1 serves as the central hub connected to two peripheral nodes 2 and 3 as illustrated by figure \ref{fg:non-unique}. Examine the walk $w_{13} = (1,3,1,2,1,3)$ from node 1 to node 3 (dash arrows in the figure). This walk can be decomposed as the concatenation of $(1,3)\in \mathcal{W}_{13}$ and $(3,1,2,1,3)\in \mathcal{W}_{33}$, or alternatively as $(1,3,1,2,1,3)\in \mathcal{W}_{13}$ and the trivial walk $(3)\in \mathcal{W}_{33}$. Thus, even though $\mathcal{W}_{13}=\mathcal{W}_{13}\odot \mathcal{W}_{33}$, the set $\mathcal{W}_{13}$ cannot be uniquely decomposed as the concatenation of these walk sets.

\begin{figure}[htp]
\centering
\includegraphics[scale=0.8]{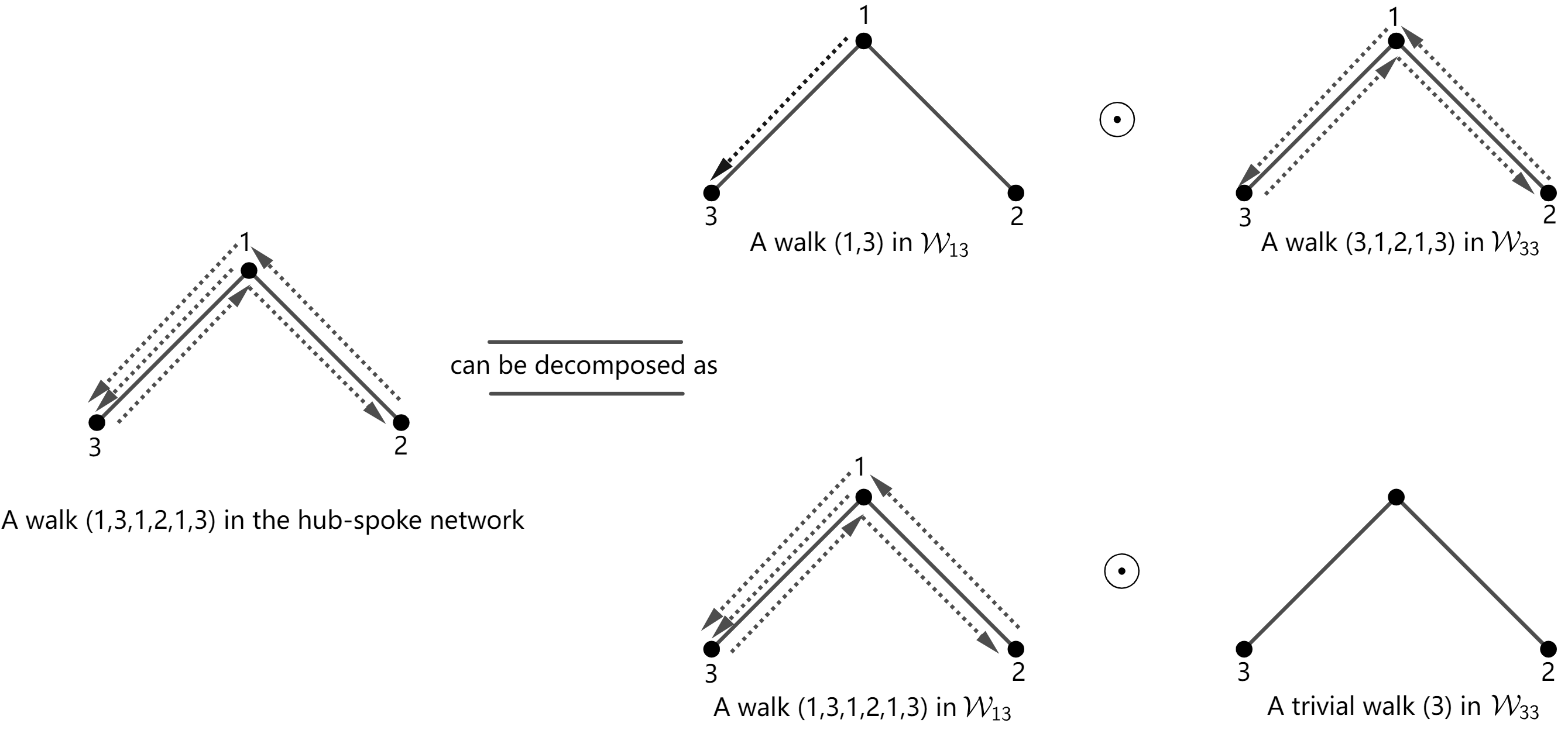}
\caption{Non-unique decomposition}
\label{fg:non-unique}
\end{figure}

Unique decomposition is achieved when each walk in the concatenated set can be formed in exactly one way and the concatenation of any pair of walks is non-empty, establishing a one-to-one correspondence between walks in $\mathcal{W}_{1}\odot \mathcal{W}_{2}$ and pairs of walks from the Cartesian product $\mathcal{W}_{1}\times \mathcal{W}_{2}$. This property ensures that we can enumerate walks in the concatenated set by analyzing the component sets separately and forms the foundation for deriving relationships among generating functions. We now present a systematic approach to uniquely decompose walks based on their traversal of specified nodes or links.

\begin{lemma}[First Passage Decomposition]
\label{pro:walk decomposition} 
For any nodes $i,j\in N$, any set of nodes $A\subseteq N$, and any set of links $\mathcal{L}\subseteq \mathcal{G}$: 
\begin{eqnarray}
\mathcal{W}_{ij}(A) &=&\bigcup_{l\in A}\{\mathcal{W}_{il}(\lnot A)\setminus \{(i)\}\odot \mathcal{W}_{lj}\setminus \{(j)\}\}; \label{eq:FPD1} \\
\mathcal{W}_{ij}(\mathcal{L}) &=&\bigcup_{(l,k)\in \mathcal{L}}\{\mathcal{W}_{il}(\lnot \mathcal{L})\odot (l,k)\odot \mathcal{W}_{kj}\}. \label{eq:FPD2}
\end{eqnarray}
Moreover, these decompositions are unique.
\end{lemma}

Lemma \ref{pro:walk decomposition} decomposes walks based on their first passage through either nodes in set $A$ or links in set $\mathcal{L}$. For walks in $\mathcal{W}_{ij}(A)$, let $l\in A$ be the first node in set $A$ encountered. The walk uniquely decomposes into two components: a walk from $i$ to $l$ that avoids set $A$, followed by a walk from $l$ to $j$. Since $\mathcal{W}_{ij}(A)$ contains no length-zero walks and all its walks have length at least 2, we exclude length-zero walks from the concatenation factors. For walks in $\mathcal{W}_{ij}(\mathcal{L})$, a parallel three-part decomposition exists: a walk to some link in $\mathcal{L}$ using no links from $\mathcal{L}$, the link itself, and a walk from the link's endpoint to the destination.

Lemma \ref{pro:walk decomposition} shares conceptual foundations with the classical first passage decomposition in stochastic processes, though they serve different analytical purposes. Using our notation, the decomposition can be expressed as 
\begin{equation}
\mathcal{W}_{ij}=\mathcal{W}_{ij}(\lnot \{j\})\odot \mathcal{W}_{jj}\text{ when }j\neq i\text{,}  \label{eq:FPD}
\end{equation}
which characterizes all walks from $i$ to $j$ by their first arrival at $j$, enabling the study of first-passage times in applications such as threshold-crossing events (see \citealp{Norris2004}). Our decomposition instead characterizes walks that must traverse specified nodes (or links), structuring them by the first encounter with these required elements. Both
decompositions leverage the principle of partitioning walks at critical first encounters.

\begin{ex}
Consider the case where $A=N_{i}$. By equation \eqref{eq:neighbors}, $\mathcal{W}_{ij}(N_{i})=\mathcal{W}_{ij}$ when $j\notin N_{i}\cup \{i\}$. Since $\mathcal{W}_{il}(\lnot N_{i})=\{(i,l)\}$ for $l\in N_{i}$, equation \eqref{eq:FPD1} yields
\begin{equation}
\mathcal{W}_{ij}=\bigcup_{l\in N_{i}}\{(i,l)\odot \mathcal{W}_{lj}\}\text{.}  \label{eq:example}
\end{equation}
This provides a unique decomposition of $i$-$j$ walks: each walk consists of a link from $i$ to some neighbor $l$ followed by a walk from $l$ to $j$.

When $\mathcal{L}=\mathcal{L}_{i}$ is the set of $i$'s out-links and $j\neq i$, equation \eqref{eq:FPD2} yields
\begin{equation*}
\mathcal{W}_{ij}=\bigcup_{(i,l)\in \mathcal{L}_{i}}\{(i)\odot (i,l)\odot \mathcal{W}_{lj}\}=\bigcup_{(i,l)\in \mathcal{L}_{i}}\{(i,l)\odot \mathcal{W}_{lj}\}
\end{equation*}
where we use equation \eqref{eq:links}: $\mathcal{W}_{ij}(\lnot \mathcal{L}_{i})=\emptyset$ when $j\neq i$ and $\mathcal{W}_{ij}(\mathcal{L}_{i})=\mathcal{W}_{ij}$. Thus, each $i$-$j$ walk uniquely decomposes into an out-link $(i,l)$ followed by a walk from $l$ to $j$.
\end{ex}

We now introduce generating functions as an analytical framework to encode walk enumeration in networks. These functions translate our walk decomposition results into algebraic relationships, enabling the simultaneous analysis of walks across all lengths.

Any set of walks $\mathcal{W^{\prime}}\subseteq \mathcal{W}$ forms a combinatorial class with a natural size function $f_{\mathcal{W^{\prime}}}(w)=\#(w)$ measuring walk length. This leads to the following definition:

\begin{defin}
For any $\mathcal{W^{\prime}}\subseteq \mathcal{W}$, the generating function is a formal power series 
\begin{equation*}
M[[\mathcal{W^{\prime}};x]]:=\sum_{t\geq 0}|f^{-1}_{\mathcal{W^{\prime}}}(t)|\cdot x^{t}.
\end{equation*}
\end{defin}

In this generating function, each term $x^{t}$ corresponds to walks of length $t$ in the set $\mathcal{W'}$, with coefficient $|f^{-1}_{\mathcal{W'}}(t)|$ counting the number of such walks. When $\mathcal{W'}$ contains only the empty walk, $M[[\mathcal{W'};x]]=1$. The function $M[[\mathcal{W}_{ij};x]]$, known as the walk generating function, encodes all walks between nodes $i$ and $j$ in the network \citep{Godsil1992}. While traditionally used to analyze how structural interventions affect the adjacency matrix determinant \citep{Rowlinson1996}, we extend this concept to enumerate arbitrary subsets of walks, enabling analysis of how specific nodes and links contribute to connectivity between node pairs.

Applying the sum and product properties from Section \ref{sec:preliminary}, we establish the following two equalities:

\begin{enumerate}[(i)]
\item \textbf{Sum Property:} $M[[\mathcal{W}_{1}\cup \mathcal{W}_{2};x]]=M[[\mathcal{W}_{1};x]]+M[[\mathcal{W}_{2};x]]$ provided that $\mathcal{W}_{1}$ and $\mathcal{W}_{2}$ are disjoint;

\item \textbf{Product Property:} $M[[\mathcal{W}_{1}\odot \mathcal{W}_{2};x]]=M[[\mathcal{W}_{1};x]]\cdot M[[\mathcal{W}_{2};x]]$ provided that the decomposition $\mathcal{W'}=\mathcal{W}_{1}\odot \mathcal{W}_{2}$ is unique.
\end{enumerate}

The sum property follows directly from the additive nature of generating functions applied to disjoint walk sets. The product property establishes the relationship between walk concatenation and multiplication of generating functions. The uniqueness ensures a bijection between walks in the concatenated set $\mathcal{W}_{1}\odot \mathcal{W}_{2}$ and pairs in the Cartesian product $\mathcal{W}_{1}\times \mathcal{W}_{2}$. Since the length of a concatenated walk equals the sum of its components' lengths, this bijection transforms the concatenation into multiplication of generating functions.

Applying the product property to the classical first passage decomposition \eqref{eq:FPD} yields:
\begin{equation}
M[[\mathcal{W}_{ij};x]]=M[[\mathcal{W}_{ij}(\lnot\{j\});x]]\cdot M[[\mathcal{W}_{jj};x]].  \label{eq:FPDM}
\end{equation}
This aligns with \citeauthor{Zhou2015}'s (\citeyear{Zhou2015}) technical lemma (equation 12) showing that $M[[\mathcal{W}_{ij};x]]\cdot M^{-1}[[\mathcal{W}_{jj};x]]$ counts the walks from $i$ to $j$ that visit $j$ only as their terminal node.\footnote{Note that $M[[\mathcal{W}_{jj};x]]$ is invertible in the ring of formal power series as its constant term is one and therefore non-zero.  The inverse of  $M[[\mathcal{W}_{jj};x]]$ is denoted by $M^{-1}[[\mathcal{W}_{jj};x]]$, which satisfies $M[[\mathcal{W}_{jj};x]]\cdot M^{-1}[[\mathcal{W}_{jj};x]]=1$ in the ring of formal power series.}

In the following part, we fix the indeterminate $x$ and simplify notation
by writing $m_{ij},m_{ij}^{\lnot A}$ and $m_{ij}^{\lnot \mathcal{L}}$ in
place of $M[[\mathcal{W}_{ij};x]]$, $M[[\mathcal{W}_{ij}(\lnot A);x]]$ and $M[[\mathcal{W}_{ij}(\lnot \mathcal{L});x]]$.

We now present the first passage decomposition of walks in terms of generating functions. Let $\mathbf{1}_{l=i}$ denote the indicator function that takes value 1 if $l=i$ and 0 otherwise.

\begin{thm}
\label{pro:GF} For any nodes $i,j\in N$:

(i) Given any set of nodes $A\subseteq N$ 
\begin{align}
m_{ij}& =m_{ij}^{\lnot A}+m_{ij}^{A},  \label{eq:GF1} \\
m_{ij}^{A}& =\sum_{l\in A}(m_{il}^{\lnot A}-\mathbf{1}_{l=i})(m_{lj}-\mathbf{1}_{l=j}).  \label{eq:GF2}
\end{align}

(ii) Given any set of links $\mathcal{L}\subseteq \mathcal{G}$ 
\begin{align}
m_{ij}& =m_{ij}^{\lnot \mathcal{L}}+m_{ij}^{\mathcal{L}},  \label{eq:GF3} \\
m_{ij}^{\mathcal{L}}& =\sum_{(l,k)\in \mathcal{L}}m_{il}^{\lnot \mathcal{L}}\cdot m_{kj}\cdot x.  \label{eq:GF4}
\end{align}
\end{thm}

Theorem \ref{pro:GF} derives directly from applying generating functions to the decomposition of walks. The walk partition identities \eqref{eq:partition1} and \eqref{eq:partition2} translate to equations \eqref{eq:GF1} and \eqref{eq:GF3} through the additivity property of generating functions. Similarly, the first passage decomposition through node set $A$ or link set $\mathcal{L}$ becomes equations \eqref{eq:GF2} and \eqref{eq:GF4}. This algebraic translation converts a structural decomposition into a system of equations that can be solved explicitly.

\begin{ex}[Continued]
Setting $A=N_{i}$ and applying equation \eqref{eq:neighbors} (or setting $\mathcal{L}=\{(i_{0},i_{1})\in \mathcal{G}: \text{either }i_{0}=i \text{ or }i_{1}=i\}$ and applying equation \eqref{eq:links}), the generating function yields for any nodes $i,j\in N$:
\begin{equation}
m_{ij}=\mathbf{1}_{j=i}+\sum_{l\in N_{i}}m_{lj}\cdot x=\left[(\mathbf{I}-x\mathbf{G})^{-1}\right]_{ij}.  \label{eq:KB2}
\end{equation}

This recursion shows that the generating function for $i$-$j$ walks equals $\mathbf{1}_{j=i}$ plus the sum of generating functions for $l$-$j$ walks over $i$'s neighbors $l$, each weighted by $x$. Equation \eqref{eq:KB2} expresses $m_{ij}$ as the $(i,j)$-entry of $(\mathbf{I}-x\mathbf{G})^{-1}$, which expands as the formal power series $\sum_{k=0}^{\infty}x^{k}\mathbf{G}^{k}$ in $x$ with matrix coefficients.

This establishes a direct connection to Katz-Bonacich centrality, where each node's importance depends on its neighbors' importance. When $x\in(0,\frac{1}{\lambda_{\max}(\frac{\mathbf{G}+\mathbf{G}^{\prime}}{2})})$, where $\lambda_{\max}(\frac{\mathbf{G}+\mathbf{G}^{\prime}}{2})$ is the spectral radius of the symmetrized adjacency matrix, $\sum_{j}m_{ij}$ converges to node $i$'s Katz-Bonacich centrality.
\end{ex}

Theorem \ref{pro:GF} transforms the structural decomposition into a system of equations that can be solved explicitly. It provides a complete characterization of the relationships between generating functions $m_{ij}$, $m_{ij}^{A}$, and $m_{ij}^{\lnot A}$ (and analogously for link sets $\mathcal{L}$). Each function encodes the number of walks from $i$ to $j$ of
various lengths under different restrictions. These relationships form a system of two equations with three unknowns, allowing us to determine the other two once given one of the three. In practice, if we assume $m_{ij}$ is known and express $m_{ij}^{A}$ (or $m_{ij}^{\mathcal{L}}$) in terms of $m_{ij}$ by solving the resulting linear system, we capture the impacts of removing a set of nodes (or links) on the $i$-$j$ walks, as the term $m_{ij}^{A}$ represents walks that must pass through nodes in $A$. Conversely, if we assume $m_{ij}^{\lnot \mathcal{L}}$ is known and express $m_{ij}$ as a function of $m_{ij}^{\lnot \mathcal{L}}$, we capture the effects of adding links $\mathcal{L}$, since $m_{ij}^{\lnot \mathcal{L}}$ represents walks in the network without links $\mathcal{L}$ while $m_{ij}$
represents walks in the full network. In summary, the framework provided by Theorem \ref{pro:GF} enables analysis of how the choice of sets $A$ and $\mathcal{L}$ affects walks, which has economic applications as we demonstrate in the next section.

\section{Applications}
\label{sec:app}

Building on our formal power series framework, we now analyze how structural changes redistribute walks throughout the network. This analysis has immediate implications for policy design in economic networks, where targeted interventions can yield significant welfare improvements. Before proceeding, we establish several notational conventions to streamline the presentation.

Let $\mathbf{M}=\left(m_{ij}\right)$ denote the matrix of generating functions whose elements are formal power series. Similarly, define matrices $\mathbf{M}^{A}=(m_{ij}^{A})$, $\mathbf{M}^{\lnot A}=(m_{ij}^{\lnot A})$, $\mathbf{M}^{\mathcal{L}}=(m_{ij}^{\mathcal{L}})$, and $\mathbf{M}^{\lnot \mathcal{L}}=(m_{ij}^{\lnot \mathcal{L}})$ for any node set $A$ and link set $\mathcal{L}$. We use $\mathbf{M}^{-1}$ to denote the inverse of generating function matrix $\mathbf{M}$ such that $\mathbf{M}^{-1}\mathbf{M=I}$.\footnote{
The generating function matrix $\mathbf{M}\left[\left[x\right]\right]$ is invertible in the ring of formal power series matrices because its constant term, $\mathbf{M} = \mathbf{I}$, is invertible.} For notational convenience with submatrices, given an $n\times n$ matrix $\mathbf{Q}$ and a set $A\subseteq N$, we use $\mathbf{Q}_{AA}$ to denote the submatrix corresponding to rows and columns indexed by $A$. After appropriate node reordering, we can represent $\mathbf{Q}$ as the block matrix 
\begin{equation*}
\left[ 
\begin{array}{cc}
\mathbf{Q}_{A^{C}A^{C}} & \mathbf{Q}_{A^{C}A} \\ 
\mathbf{Q}_{AA^{C}} & \mathbf{Q}_{AA}%
\end{array}%
\right] \text{,}
\end{equation*}%
where $A^{C}$ denotes the complement of set $A$.

\subsection{Structural Interventions}

\noindent\textbf{Nodes Removal.} The concept of key node in networks was first introduced by \cite{Ballester2006} in their study of criminal networks, where they sought to identify the individual whose removal would most effectively reduce criminal activities. Mathematically, this translates to finding the node whose removal maximally decreases the aggregate Katz-Bonacich centrality. 

In this spirit, we first analyze how removing a set of nodes $A$ affects walks between the remaining nodes in the network. These residual walks are characterized by $\mathbf{M}_{A^{C}A^{C}}^{\lnot A}$, which represents walks between node pairs in $A^{C}$ that never pass through nodes in $A$. The following corollary applies Theorem \ref{pro:GF} to express $\mathbf{M}_{A^{C}A^{C}}^{\lnot A}$ in terms of walks in the original network. 

Let $\dot{b}_{i}=\sum_{j\in N}m_{ij}$ and $\mathring{b}_{i}=\sum_{j\in N}m_{ji}$ denote the outgoing and incoming generating functions for node $i$, respectively. In undirected networks, these two generating functions are identical. When these series converge, they yield $i$'s outgoing and incoming Katz-Bonacich centrality in the network. We denote the corresponding vectors as $\mathbf{\dot{b}}=(\dot{b}_{i})$ and $\mathbf{\mathring{b}}=(\mathring{b}_{i})$.
\begin{pro}
\label{cor:nodes} For any node set $A\subseteq N$, we have 
\begin{equation}
\mathbf{M}_{A^{C}A^{C}}^{\lnot A}=\mathbf{M}_{A^{C}A^{C}}-\mathbf{M}_{A^{C}A}(\mathbf{M}_{AA})^{-1}\mathbf{M}_{AA^{C}}\text{.}  \label{eq:nodes1}
\end{equation}
In particular, the decrease in the number of walks caused by removing nodes $A$ is:
\begin{equation}
\sum_{i,j\in N}m_{ij}-\sum_{i,j\in A^{C}}m_{ij}^{\lnot A}=\mathbf{\mathring{b}}_{A}^{\prime }(\mathbf{M}_{AA})^{-1}\mathbf{\dot{b}}_{A}  \label{eq:nodes2}
\end{equation}
\end{pro}

Equation (\ref{eq:nodes1}) emerges directly from the matrix system in Theorem \ref{pro:GF}(i): 
\begin{eqnarray*}
\mathbf{M}_{A^{C}A^{C}} &=&\mathbf{M}_{A^{C}A^{C}}^{\lnot A}+\mathbf{M}_{A^{C}A^{C}}^{A}; \\
\mathbf{M}_{A^{C}A^{C}}^{A} &=&\mathbf{M}_{A^{C}A}^{\lnot A}\mathbf{M}_{AA^{C}}\text{,}
\end{eqnarray*}
with the substitution of $\mathbf{M}_{A^{C}A}^{\lnot A}=\mathbf{M}_{A^{C}A}\left( \mathbf{M}_{AA}\right) ^{-1}$ (from equation \ref{eq:FPDM}). It characterizes the impact on walks between nodes in $A^{C}$ when removing the set $A$ of nodes. In particular, taking $A=\{i\}$ as a singleton, then for any $j\neq i\neq k$, equation (\ref{eq:nodes1}) yields:
\begin{equation*}
m_{jk}-m_{jk}^{\lnot \{i\}}=m_{ji}\cdot m_{ii}^{-1}\cdot m_{ik}\text{.}
\end{equation*}
This equation characterizes how node $i$ influences $j$-$k$ walks of various lengths. It generalizes Lemma 1 of \cite{Ballester2006} from convergent values to formal power series, providing a tool for analyzing node $i$'s role as a walk intermediator.

From equation (\ref{eq:nodes1}), we know that the reduction in walks between nodes in $A^{C}$ is given by $\mathbf{M}_{A^{C}A}(\mathbf{M}_{AA})^{-1}\mathbf{M}_{AA^{C}}$. The total reduction in walks across the network can be expressed as: 
\begin{equation*}
\sum_{i,j\in N}m_{ij}-\sum_{i,j\in A^{C}}m_{ij}^{\lnot A}=\underbrace{\mathbf{1}^{\prime }\mathbf{M}_{A^{C}A}(\mathbf{M}_{AA})^{-1}\mathbf{M}_{AA^{C}}\mathbf{1}}_{\text{reduction in walks between $A^{C}$ nodes}}+\underbrace{\mathbf{1}^{\prime }\mathbf{M}_{A^{C}A}\mathbf{1}}_{\text{walks to removed nodes}}+\underbrace{\mathbf{1}^{\prime }\mathbf{M}_{AA}\mathbf{1}+\mathbf{1}^{\prime }\mathbf{M}_{AA^{C}}\mathbf{1}}_{\text{walks from removed nodes}}\text{.}
\end{equation*}
Simplifying these terms yields equation (\ref{eq:nodes2}). The converged value of $\sum_{i,j\in N}m_{ij}-\sum_{i,j\in A^{C}}m_{ij}^{\lnot A}$ is called the group intercentrality of $A$ introduced by \cite{Ballester2010} and further explored by \cite{sun2023}.\footnote{\cite{sun2023} derive the same expression for group intercentrality as equation (\ref{eq:nodes2}) through a different method. Specifically, \cite{sun2023} derive the expression by the equivalent translation between removing nodes and decreasing nodes' characteristics in a network game. Here, we derive the expression by analyzing the change in walks through generating functions. This approach also offers a walk-based explanation of group intercentrality, which is absent in \cite{sun2023}.} It measures the impact of removing group $A$ from a criminal network on overall delinquency in \cite{Ballester2010}. Equation (\ref{eq:nodes2}) generalizes \cite{Ballester2006}'s intercentrality measure to multi-node sets and directed networks. This characterization reveals that a node set's importance depends on three crucial factors: the walks ending at this set ($\mathbf{\mathring{b}}_{A}$), the walks starting from this set ($\mathbf{\dot{b}}_{A}$)—which are identical in undirected networks—and the internal connection patterns ($\mathbf{M}_{AA}$) within the set.

The term $(\mathbf{M}_{AA})^{-1}$ in the group intercentrality expression (\ref{eq:nodes2}) stems from the generating function of the classical first passage decomposition $\mathbf{M}_{A^{C}A}^{\lnot A}=\mathbf{M}_{A^{C}A}\left( \mathbf{M}_{AA}\right) ^{-1}$. The stronger the internal connections within a group, the less impact its removal has on walks starting from the remaining nodes in $A^{C}$. This principle becomes clear in regular undirected networks, where all nodes have identical Katz-Bonacich centrality. In such networks, for two node sets $A$ and $A^{\prime}$ of equal size, the decrease in the total number of walks caused by removing $A$ is larger than that caused by removing $A^{\prime}$ whenever $\mathbf{M}_{AA}\leq \mathbf{M}_{A^{\prime}A^{\prime}}$ element-wise. Thus, less internally connected groups often exert greater network-wide influence. This formula also explains the failure of greedy algorithms that sequentially select the player maximizing individual intercentrality at each step, as noted by \cite{Ballester2010}. Such greedy approaches ignore the crucial impact of internal group connections on the network's walk structure.

\bigskip

\noindent\textbf{Link modifications.} Rewriting Theorem \ref{pro:GF}(ii) in matrix form yields:
\begin{eqnarray*}
\mathbf{M} &=&\mathbf{M}^{\lnot \mathcal{L}}+\mathbf{M}^{\mathcal{L}}; \\
\mathbf{M}^{\mathcal{L}} &=&x\mathbf{M}^{\lnot \mathcal{L}}\mathbf{LM},
\end{eqnarray*}
where $\mathbf{L}$ is an $n\times n$ matrix with $ij$-th element being $1$ if $(i,j) \in \mathcal{L}$ and $0$ otherwise. Treating $\mathbf{M}$ as known and solving this equation system, we obtain:
\begin{equation}
\mathbf{M}^{\lnot \mathcal{L}}=\mathbf{M}(\mathbf{I}+x\mathbf{LM})^{-1}\text{.}  \label{eq:links1}
\end{equation}

Equation (\ref{eq:links1}) characterizes walks that never traverse links in $\mathcal{L}$ in terms of the generating function matrix $\mathbf{M}$. Equivalently, it represents walks in the network after removing links $\mathcal{L}$. Conversely, treating $\mathbf{M}^{\lnot \mathcal{L}}$ as known and solving the equation system, we obtain:
\begin{equation}
\mathbf{M}=(\mathbf{I}-x\mathbf{M}^{\lnot \mathcal{L}}\mathbf{L})^{-1}\mathbf{M}^{\lnot \mathcal{L}}\text{.}  \label{eq:links2}
\end{equation}

Equation (\ref{eq:links2}) characterizes walks in the full network in terms of walks that never traverse $\mathcal{L}$. This formulation captures the impact of adding links $\mathcal{L}$, where $\mathbf{M}^{\lnot \mathcal{L}}$ represents walks in the original network and $\mathbf{M}$ represents walks after adding links $\mathcal{L}$. By applying equations (\ref{eq:links1}) and (\ref{eq:links2}) iteratively, we can analyze the impact of complex link modifications involving both additions and removals.

Formally, consider a structural intervention that removes links in set $\mathcal{L}_{-}\subseteq \mathcal{G}$ and adds links $\mathcal{L}_{+}\subseteq \mathcal{G}^C$ (where $\mathcal{G}^C$ denotes the complement of $\mathcal{G}$). Let $\mathbf{L}_{-}$ and $\mathbf{L}_{+}$ denote the corresponding matrices such that the $ij$-th element equals $1$ in $\mathbf{L}_{-}$ whenever $(i,j) \in \mathcal{L}_{-}$ and equals $1$ in $\mathbf{L}_{+}$ whenever $(i,j) \in \mathcal{L}_{+}$.

\begin{pro}
\label{cor:link} The generating function matrix for walks in the post-intervention network is: 
\begin{equation*}
\mathbf{M}(\mathcal{L}_{-},\mathcal{L}_{+})=\underbrace{\mathbf{M}}_{\text{original walks}}+\underbrace{x\mathbf{M}(\mathbf{L}_{+}-\mathbf{L}_{-})(\mathbf{I}-x(\mathbf{L}_{+}-\mathbf{L}_{-})\mathbf{M})^{-1}\mathbf{M}}_{:=\Delta \mathbf{M}(\mathcal{L}_{-},\mathcal{L}_{+})\text{ change in walks}}\text{.}
\end{equation*}
\end{pro}

Proposition \ref{cor:link} characterizes the impact of structural interventions on walks between all node pairs. The change in walks has a clear walk-counting interpretation:
\begin{align*}
\Delta \mathbf{M}(\mathcal{L}_{-},\mathcal{L}_{+}) = &\underbrace{x\mathbf{M}(\mathbf{L}_{+}-\mathbf{L}_{-})\mathbf{M}}_{\text{passing the changed links once}} + \underbrace{x^{2}\mathbf{M}(\mathbf{L}_{+}-\mathbf{L}_{-})\mathbf{M}(\mathbf{L}_{+}-\mathbf{L}_{-})\mathbf{M}}_{\text{passing the changed links twice}} \\
&+ \underbrace{\mathbf{...}}_{\text{passing the changed links multiple times}}
\end{align*}

When considering the change in the total number of walks in the network:
\begin{equation}
\mathbf{1}^{\prime }\Delta \mathbf{M}(\mathcal{L}_{-},\mathcal{L}_{+})\mathbf{1}=x\mathbf{\mathring{b}}^{\prime }(\mathbf{L}_{+}-\mathbf{L}_{-})(\mathbf{I}-x(\mathbf{L}_{+}-\mathbf{L}_{-})\mathbf{M})^{-1}\mathbf{\dot{b}}\text{,}  \label{eq:keylink2}
\end{equation}
which is consistent with Proposition 1 in \cite{sun2023} (with different notation). In particular, when $x\in (0,\frac{1}{\lambda _{\max }(\frac{\mathbf{G}+\mathbf{G}^{\prime }}{2})})$ and considering convergent values of the generating functions, the link set $\mathcal{L}^{\ast}$ that maximizes $\mathbf{1}^{\prime }\Delta \mathbf{M}(\mathcal{L},\emptyset )\mathbf{1}$ identifies the key links whose removal would result in the largest decrease in aggregate Katz-Bonacich centrality. This index ranks critical channels in criminal networks, where targeted removals would maximally disrupt criminal activities.

The dimension of $(\mathbf{L}_{+}-\mathbf{L}_{-})(\mathbf{I}-x(\mathbf{L}_{+}-\mathbf{L}_{-})\mathbf{M})^{-1}$ is determined by the rank of $\mathbf{L}_{+}-\mathbf{L}_{-}$. When few nodes are involved in link modifications—such as single link addition or deletion—the rank of $\mathbf{L}_{+}-\mathbf{L}_{-}$ is low, making the calculation of aggregate walk change $\mathbf{1}^{\prime }\Delta \mathbf{M}(\mathcal{L}_{-},\mathcal{L}_{+})\mathbf{1}$ computationally efficient. For example, adding a single undirected link $(i,j)$ to an undirected network results in a rank-two matrix $\mathbf{L}_{+}=\mathbf{E}_{ij}$ where only the $(i,j)$-th and $(j,i)$-th elements are 1 and all others are 0. The change in total walks in the network is given by:
\begin{equation*}
\mathbf{1}^{\prime }\Delta \mathbf{M}(\emptyset ,\{(i,j)\})\mathbf{1}=x\mathbf{\mathring{b}}_{\{i,j\}}^{\prime }\mathbf{E}_{ij}(\mathbf{I}+x\mathbf{E}_{ij}\mathbf{M}_{\{i,j\}\{i,j\}})^{-1}\mathbf{\dot{b}}_{\{i,j\}}\text{,}
\end{equation*}
which can be computed by calculating the inverse of the rank-two matrix $(\mathbf{I}+x\mathbf{E}_{ij}\mathbf{M}_{\{i,j\}\{i,j\}})$. Similar analysis applies to single link $(k,l)$ deletion. Consequently, the changes in the total number of walks caused by adding a single link $(i,j)$ and deleting a single link $(k,l)$ are given by: 
\begin{eqnarray*}
\mathbf{1}^{\prime }\Delta \mathbf{M}(\emptyset ,\{(i,j)\})\mathbf{1} &=&x\frac{xm_{jj}b_{i}^{2}+xm_{ii}b_{j}^{2}+2(1-xm_{ij})b_{i}b_{j}}{(1-xm_{ij})^{2}-x^{2}m_{ii}m_{jj}} \\
\mathbf{1}^{\prime }\Delta \mathbf{M}(\{(k,l)\},\emptyset )\mathbf{1} &=&-x\frac{xm_{ll}b_{k}^{2}+xm_{kk}b_{l}^{2}-2(1+xm_{kl})b_{k}b_{l}}{(1-xm_{kl})^{2}-x^{2}m_{kk}m_{ll}}
\end{eqnarray*}

\noindent These equations characterize the impact of adding and deleting links on aggregate Katz-Bonacich centrality, as shown by \cite{Ballester2010} and \cite{sun2023}.

\subsection{Information Transmission}

Consider a model of information diffusion on a symmetric, unweighted social network $\mathbf{G}$ studied by \cite{banerjee2013diffusion}, \cite{Banerjee2019using}, and \cite{Cruz2017politician}. Information is transmitted in discrete iterations: at period 0, a sender possesses initial information; in subsequent periods, informed agents independently transmit this information to their network neighbors with probability $\delta \in (0,\frac{1}{\lambda_{\max}(\mathbf{G})})$. \cite{Bramoulle2018} enhanced this model by introducing an information recipient, motivated by political intermediation scenarios where citizens' requests are transmitted to politicians.

The diffusion process operates under three key assumptions:

\begin{ass} ~ ~

(i) Information transmission occurs only from agents informed in period $t-1$;

(ii) An agent receiving information from $s$ distinct sources in period $t-1$ initiates $s$ independent transmissions in period $t$;

(iii) The sender or receiver or both do not retransmit information (no re-transmission).

\end{ass}

We generalize \cite{Bramoulle2018}'s target centrality concept from individual nodes to a group of nodes. This extension addresses scenarios like state-level political representation, where citizens' requests may reach multiple representatives. Let $\xright{n}_{iA}$, $\xleft{n}_{iA}$, and $\xboth{n}_{iA}$ denote the expected number of times group $A$ receives a message from sender $i$ under different non-retransmission constraints: sender-only, recipient-only, and both sender-and-recipient, respectively. Formally, for any node set $A$, let $\mathcal{W}_{ij}^{t}(\lnot A)=\{w_{ij}\in \mathcal{W}_{ij}(\lnot A):\#(w_{ij})=t\}$ represent the set of length-$t$ walks from $i$ to $j$ avoiding nodes in $A$. The expected reception counts are then: 
\begin{align*}
\xright{n}_{iA} &:=\sum_{j\in A}\sum_{t=0}^{\infty}\delta^{t}|\mathcal{W}_{ij}^{t}(\lnot \{i\})|, \\
\xleft{n}_{iA}&:=\sum_{j\in A}\sum_{t=0}^{\infty}\delta^{t}|\mathcal{W}_{ij}^{t}(\lnot A)|, \\
\xboth{n}_{iA}&:=\sum_{j\in A}\sum_{t=0}^{\infty}\delta^{t}|\mathcal{W}_{ij}^{t}(\lnot (A\cup \{i\}))|.
\end{align*}

In this context, the non-retransmission group $A$ acts as a set of absorbing states where messages can enter but not exit. The expected number of times a receiver receives a message naturally relates to the convergent values of $\mathbf{M}^{\lnot A}$, which captures walks avoiding the non-retransmission group $A$. In what follows, we use $m_{ij}^{\lnot A}=M\left[\left[\mathcal{W}_{ij}(\lnot A);\delta\right]\right]$ to denote the corresponding convergent value of the generating function when $\delta\in (0,\frac{1}{\lambda_{\max}(\mathbf{G})})$, as we are concerned with the total number of message receptions rather than the formal power series.

\begin{proposition}
\label{pro:target} For any node set $A\subseteq N$ and node $i\notin A$, the expected message reception counts are: 
\begin{equation*}
\xright{n}_{iA}=\sum_{j\in A}m_{ij}^{\lnot\{i\}}; \quad 
\xleft{n}_{iA}=\sum_{j\in A}m_{ij}^{\lnot A};\quad 
\xboth{n}_{iA}=\sum_{j\in A}m_{ij}^{\lnot(A\cup\{i\})}.
\end{equation*}
For the special case where $A=\{j\}$: 
\begin{equation*}
\xright{n}_{ij}=\frac{m_{ij}}{m_{ii}}; \quad 
\xleft{n}_{ij}=\frac{m_{ij}}{m_{jj}}; \quad 
\xboth{n}_{ij}=\frac{m_{ij}}{m_{ii}m_{jj}-(m_{ij})^2}.
\end{equation*}
\end{proposition}

Proposition \ref{pro:target} characterizes group target centrality using elements of $\mathbf{M}^{A}$, which can be derived from the matrix form of Theorem \ref{pro:GF}(i). For a group $A$, we define three aggregate measures of target centrality, corresponding to the expected number of times group $A$ receives information from all other individuals under different non-retransmission constraints:

\begin{equation*}
\xright{n}_{A}=\sum_{i\notin A}\xright{n}_{iA}, \quad 
\xleft{n}_{A}=\sum_{i\notin A}\xleft{n}_{iA}, \quad 
\xboth{n}_{A}=\sum_{i\notin A}\xboth{n}_{iA}
\end{equation*}

Combining Proposition \ref{pro:target} and the expression of $\mathbf{M}^{A}$ derived from Theorem \ref{pro:GF}(i) yields explicit formulas: 
\begin{align*}
\xright{n}_{A}&=\sum_{i\notin A}\sum_{j\in A}\frac{m_{ij}}{m_{ii}}, \\
\xleft{n}_{A}&=\mathbf{b}_{A}(\mathbf{M}_{AA})^{-1}\mathbf{1}-|A|, \\
\xboth{n}_{A} &=2|A|-\sum_{i\notin A}\left(\mathbf{1}^{\prime}(\mathbf{M}_{A\cup \{i\},A\cup \{i\}})^{-1}\mathbf{1}-[(\mathbf{M}_{A\cup \{i\},A\cup \{i\}})^{-1}]_{ii}\right).
\end{align*}

For single-node targeting ($A=\{j\}$), the formulas demonstrate that expected reception counts increase with the proximity measure $m_{ij}$ between nodes while decreasing with the self-loops ($m_{ii}$ and $m_{jj}$) of nodes excluded from retransmission. Similar to our discussion of group intercentrality in the previous section, the greedy algorithm fails to identify the optimal target group because internal connectivity $\mathbf{M}_{AA}$ significantly influences the value of target centrality.

Our characterization offers advantages over \cite{Bramoulle2018}'s approach in both computational efficiency and analytical power. Computationally, we express all target centrality measures for single nodes using only the Leontief inverse matrix $\mathbf{M}=(\mathbf{I}-\delta \mathbf{G})^{-1}$, eliminating the need for the multiple additional matrices (including the modified adjacency matrices $\mathbf{G}_{-i}$) required in their framework. As \cite{Bramoulle2018} states: “Under no retransmission by both sender and target, it (target centrality) involves the inverses of one matrix of size $n$ by $n$ and $n$ matrices of size $n-1$ by $n-1$.” In contrast, our formulas in Proposition \ref{pro:target} demonstrate that only the elements of a single $n \times n$ inverse matrix are sufficient to characterize target centrality, avoiding the computational burden of inverting $n$ additional matrices of size $(n-1) \times (n-1)$. Analytically, our representation enables comparative static analysis of how both node and link modifications affect information diffusion patterns throughout the network.

Building on \cite{Ballester2006}'s concept of key players in delinquent networks, we examine a crucial question in information diffusion: who are the key information intermediaries? Our walk-counting approach enables assessment of how removing single nodes affects information flows between pairs of nodes. We focus on the case where both sender and target nodes do not retransmit, thus $\xboth{n}_{ij}=m_{ij}^{\lnot \{i,j\}}$ represents the expected number of times $j$ receives $i$'s request.

\begin{definition}
\label{defn:intermediary-centrality} The key intermediary for sender-target pair $(i,j)$ is 
\begin{equation*}
k^{int}\in \arg\max_{k\notin \{i,j\}}\left(\xboth{n}_{ij}-m_{ij}^{\lnot \{i,j,k\}}\right)
\end{equation*}
or equivalently, 
\begin{equation*}
k^{int}\in \arg\min_{k\notin \{i,j\}}m_{ij}^{\lnot\{i,j,k\}}.
\end{equation*}
\end{definition}

The intermediary index of node $k$ equals the difference between the original expected reception counts and the expected counts after removing node $k$. The key intermediary thus represents the node whose blocked retransmission most effectively impedes information flow from $i$ to $j$.

\begin{corr}
\label{corr:component-intermediary} For any node set $\{i,j,k\}\in N$: 
\begin{align*}
m_{ij}^{\lnot\{i,j,k\}} &= m_{ij}-\frac{m_{ij}m_{kk}-m_{ik}m_{jk}}{m_{ii}m_{jj}m_{kk}+2m_{ij}m_{ik}m_{jk}-m_{ii}m_{jk}^2-m_{ij}^2m_{kk}-m_{ik}^2m_{jj}} \\
&= m_{ij}-\frac{m_{ij}-m_{ij}^{\lnot\{k\}}}{m_{ii}m_{jj}+2m_{ij}m_{ij}^{\lnot\{k\}}-m_{ii}m_{jj}^{\lnot\{k\}}-m_{ij}^2-m_{jj}m_{ii}^{\lnot\{k\}}}.
\end{align*}
\end{corr}

This result, derived directly from the expression of $\mathbf{M}^{\lnot A}$ with $A=\{i,j,k\}$, reveals key insights about intermediary importance. A node $k$ becomes more critical as an intermediary for pair $(i,j)$ when it plays a stronger role in connecting $i$ and $j$ but has less influence on their self-loops. Note that if we allow for retransmission by both the sender and the target, then the key intermediary is exactly the node which maximizes $m_{ij}^{\lnot\{k\}}$. However, with no retransmission, we must account for $k$'s role in walk facilitation differently. Given fixed $m_{ij}^{\lnot\{k\}}$, the most influential intermediary minimizes importance in sender and target self-loops, as these represent redundant retransmissions.

\begin{ex}
Consider the 15-node network illustrated in Figure \ref{fg:2}, where we analyze information flow between sender $i$ and recipient $j$ with discount factor $\delta=0.18$. Table 1 compares how different nodes $k$ affect the sender-target pair under both retransmission and no-retransmission scenarios.\footnote{The table presents six representative nodes. Other nodes in the network are structurally equivalent to one of these representatives regarding their impact on the $(i,j)$ pair.}

\begin{figure}[htp]
\centering
\includegraphics[scale=0.6]{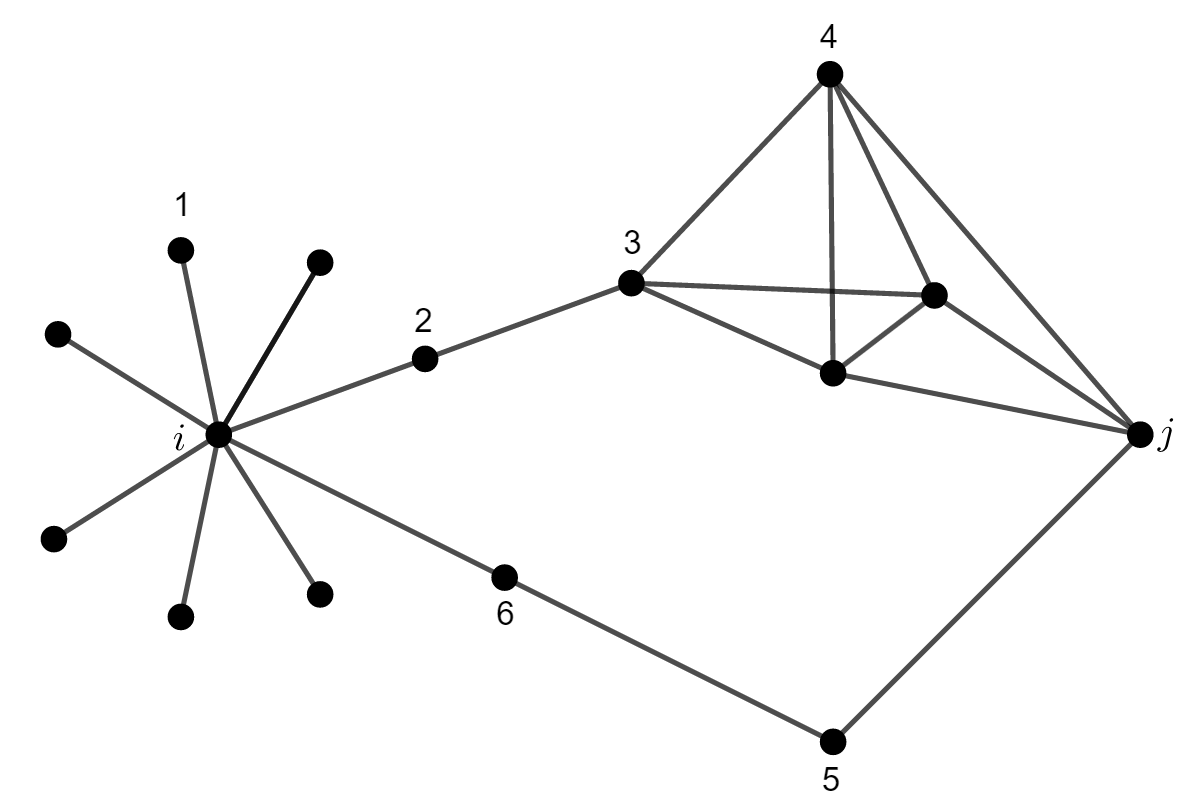}
\caption{Key intermediaries for information flow from $i$ to $j$}
\label{fg:2}
\end{figure}

\begin{center}
\begin{tabular}{|c|c|c|}
\hline
Node $k$ & $m_{ij}^{\{k\}}$ & $m_{ij}^{\{i,j,k\}}$ \\ \hline
1 & 0.001385 & 0.008714 \\ \hline
2 & $\mathbf{0.011554}^{\ast}$ & $\mathbf{0.014747}^{\ast}$ \\ \hline
3 & 0.011433 & $\mathbf{0.014747}^{\ast}$ \\ \hline
4 & 0.007506 & 0.011865 \\ \hline
5 & 0.011488 & 0.014742 \\ \hline
6 & 0.011542 & 0.014742 \\ \hline
\end{tabular}

Table 1: Intermediary importance measures
\end{center}

The analysis reveals distinct key intermediaries under different transmission scenarios. With retransmission allowed, node 2 emerges as the most critical intermediary by maximizing $m_{ij}^{\lnot\{k\}}$, indicating it facilitates the most length-discounted walks between $i$ and $j$. However, when neither $i$ nor $j$ can retransmit, both nodes 2 and 3 become equally important key intermediaries by maximizing $m_{ij}^{\lnot\{i,j,k\}}$.

A noteworthy comparison between nodes 6 and 3 demonstrates that intermediary importance isn't monotonic in $m_{ij}^{\lnot\{k\}}$. While node 6 facilitates more length-discounted walks ($m_{ij}^{\lnot\{6\}} > m_{ij}^{\lnot\{3\}}$), node 3's removal more significantly reduces expected transmission frequency under no-retransmission ($m_{ij}^{\lnot\{i,j,3\}} > m_{ij}^{\lnot\{i,j,6\}}$).
\end{ex}

\subsection{Optimal Link Construction}

Consider two nodes $i,j$ in an undirected network $\mathbf{G}$. Let $T\subseteq N\setminus (N_i \cup N_j)$ be a set of nodes, none of which is a neighbor of either $i$ or $j$. We compare the walks generated by connecting $i$ with all nodes in $T$ versus connecting $j$ with all nodes in $T$. This comparison provides valuable insights for network design problems where we must determine the optimal allocation of a fixed number of links.

We denote by $\mathbf{\hat{G}=G+}\sum_{k\in T}\mathbf{E}_{ik}$ and $\mathbf{\mathring{G}=G+}\sum_{k\in T}\mathbf{E}_{jk}$ the networks obtained from $\mathbf{G}$ by connecting $i$ with $T$ and $j$ with $T$, respectively. Furthermore, $\mathcal{\hat{W}}$ and $\mathcal{\mathring{W}}$ represent the sets of all walks in these two networks. For two generating functions $M[[\mathcal{W^{\prime}};x]]=\sum_{t\geq 0}|f_{\mathcal{W^{\prime}}}^{-1}(t)|\cdot x^{t}$ and $M[[\mathcal{W^{\prime\prime}};x]]=\sum_{t\geq 0}|f_{\mathcal{W^{\prime\prime}}}^{-1}(t)|\cdot x^{t}$, we say $M[[\mathcal{W^{\prime}};x]]$ dominates $M[[\mathcal{W^{\prime\prime}};x]]$ in the ring of formal power series, denoted by $M[[\mathcal{W^{\prime}};x]]\gtrsim M[[\mathcal{W^{\prime\prime}};x]]$, if and only if $|f_{\mathcal{W^{\prime}}}^{-1}(t)|\geq |f_{\mathcal{W^{\prime\prime}}}^{-1}(t)|$ for all $t$. Thus, two generating functions are comparable only when the coefficients of one generating function are pointwise greater than or equal to those of the other.

\begin{pro}
\label{pro:belhaj}
In $\mathbf{G}$, when $m_{ii}\gtrsim m_{jj}$ and $m_{il}\gtrsim m_{jl}$ for any $l\notin \{i,j\}$, we have 
\begin{equation*}
M[[\mathcal{\hat{W}};x]]\gtrsim M[[\mathcal{\mathring{W}};x]]\text{.}
\end{equation*}
In particular, $M[[\mathcal{\hat{W}}_{l};x]]\gtrsim M[[\mathcal{\mathring{W}}_{l};x]]$ for any $l\neq j$ and $M[[\mathcal{\hat{W}}_{i}\cup \mathcal{\hat{W}}_{j};x]]\gtrsim M[[\mathcal{\mathring{W}}_{i}\cup \mathcal{\mathring{W}}_{j};x]]$.
\end{pro}

Note that the conditions $m_{ii}\gtrsim m_{jj}$ and $m_{il}\gtrsim m_{jl}$ for any $l\notin \{i,j\}$ are equivalent to the neighborhood containment relationship $N_{i}\cup \{i\} \supseteq N_{j}\cup \{j\}$.\footnote{The fact  that neighborhood nestedness, $N_{i}\cup \{i\} \supseteq N_{j}\cup \{j\}$, implies dominance in formal power series can be shown inductively using the relationship $|f_{\mathcal{W}_{il}}^{-1}(t)|=\sum_{k \in N_{i}}|f_{\mathcal{W}_{kl}}^{-1}(t-1)|$. The converse direction follows from the length-one walks relation in the formal power series: $m_{il}\gtrsim m_{jl}, \forall l\notin \{i,j\}$.} This proposition establishes that adding links between node $i$ and set $T$ generates more walks of arbitrary length throughout the network than adding equivalent links between node $j$ and set $T$, provided that node $i$'s neighborhood forms a superset of node $j$'s neighborhood. We make three remarks about this result:

\begin{enumerate}
  \item
  This proposition strengthens \cite{Belhaj2016}'s Lemma 1, which states that shifting all neighbors of $j$ (except those who are also $i$'s neighbors) to $i$ increases aggregate Katz-Bonacich centrality. Specifically, if in network $\mathbf{\mathring{G}}$ we define $L=\{l\in N_{j}:l\notin N_{i}\}$, then $\mathbf{\hat{G}=\mathring{G}-}\sum_{l\in L}\mathbf{E}_{jl}+\sum_{l\in L}\mathbf{E}_{il}$ generates higher Katz-Bonacich centrality than $\mathbf{\mathring{G}}$. Our proposition replicates \cite{Belhaj2016}'s result when the generating functions converge by setting $T=L$ and $\mathbf{G}=\mathbf{\mathring{G}-}\sum_{l\in L}\mathbf{E}_{jl}$ in our statement.

\item   When $x$ is sufficiently small and the generating functions converge, Proposition \ref{pro:belhaj} can be extended to the converged numerical values (rather than formal power series). This result implies that for any convex and increasing function $f$, $\sum_{k\in N}f(\hat{b}_{k}) \geq \sum_{k\in N}f(\mathring{b}_{k})$, where $\hat{b}_{k}$ and $\mathring{b}_{k}$ are the converged values of $M[[\mathcal{\hat{W}}_{k};x]]$ and $M[[\mathcal{\mathring{W}}_{k};x]]$ respectively.\footnote{Here we use the fact that when $\hat{b}_{i}\geq \max\{\mathring{b}_{i},\mathring{b}_{j}\}$ and $\hat{b}_{i}+\hat{b}_{j}\geq \mathring{b}_{i}+\mathring{b}_{j}$, then $f(\hat{b}_{i})+f(\hat{b}_{j}) \geq f(\mathring{b}_{i})+f(\mathring{b}_{j})$.} This result generalizes \cite{Belhaj2016}'s Lemma 1 from the sum of Katz-Bonacich centralities to the sum of any convex and increasing function of Katz-Bonacich centrality.

\item 
When $T=\{k\}$ is a singleton, our proposition establishes that $M[[\mathcal{\hat{W}}_{k};x]]\gtrsim M[[\mathcal{\mathring{W}}_{k};x]]$. That is, connecting $k$ with $i$ results in more walks initiating from $k$ than connecting $k$ with $j$. This implies Lemma 3 in \cite{Konig2014}, which states that connecting with $i$ generates higher Katz-Bonacich centrality for node $k$ than connecting with $j$.

\end{enumerate}

This result is not entirely novel. In our companion paper \cite{sun2025}, we established it in Lemma 1 using a complex mathematical induction approach (see that paper for further discussion of the implications). Here, we provide an alternative proof using the first passage decomposition.

Let $\mathcal{L}_{i}=\{(i,k),(k,i):k\in T\}$ be the set of links between $i$ and $T$. Similarly, define $\mathcal{L}_{j}$ as the set of links between $j$ and $T$. For any $k,l\in N$, we define the set of $k$-$l$ walks in network $\mathbf{\hat{G}}$ that pass through the link set $\mathcal{L}_{i}$ exactly $n$ times as:
\begin{align*}
\mathcal{\hat{W}}_{kl}^{(n)}(\mathcal{L}_{i}):=\left\{ w\in \mathcal{W}: 
\begin{array}{l}
w=(i_{0},i_{1},\ldots,i_{K}) \text{ where } i_{0}=k, i_{K}=l, \text{ and} \\
(i_{t},i_{t+1})\in \mathcal{L}_{i} \text{ for exactly } n \text{ values of } t
\end{array}
\right\}.
\end{align*}

\noindent We analogously define $\mathcal{\hat{W}}_{k}^{(n)}(\mathcal{L}_{i})$ and $\mathcal{\mathring{W}}_{kl}^{(n)}(\mathcal{L}_{j})$ as the set of walks initiating from $k$ that pass through links in $\mathcal{L}_{i}$ exactly $n$ times, and the set of $k$-$l$ walks in network $\mathbf{\mathring{G}}$ passing through links in $\mathcal{L}_{j}$ exactly $n$ times, respectively. We can decompose:
\begin{equation*}
\mathcal{\hat{W}}_{kl}=\bigcup_{n\geq 0}\mathcal{\hat{W}}_{kl}^{(n)}(\mathcal{L}_{i})=\mathcal{\hat{W}}_{kl}(\lnot \mathcal{L}_{i})\cup\bigcup_{n\geq 1}\mathcal{\hat{W}}_{kl}^{(n)}(\mathcal{L}_{i})=\mathcal{W}_{kl}\cup\bigcup_{n\geq 1}\mathcal{\hat{W}}_{kl}^{(n)}(\mathcal{L}_{i}),
\end{equation*}
where $\mathcal{W}_{kl}$ is the set of $k$-$l$ walks in the original network $\mathbf{G}$.

Furthermore, we can uniquely decompose $\mathcal{\hat{W}}_{k}^{(n+1)}(\mathcal{L}_{i})$ as:
\begin{equation*}
\mathcal{\hat{W}}_{k}^{(n+1)}(\mathcal{L}_{i})=\bigcup_{l\in T}\{\mathcal{\hat{W}}_{ki}(\lnot \mathcal{L}_{i})\odot (i,l)\odot \mathcal{\hat{W}}_{l}^{(n)}(\mathcal{L}_{i})\cup \mathcal{\hat{W}}_{kl}(\lnot \mathcal{L}_{i})\odot (l,i)\odot \mathcal{\hat{W}}_{i}^{(n)}(\mathcal{L}_{i})\}.
\end{equation*}

\noindent That is, each walk in $\mathcal{\hat{W}}_{k}^{(n+1)}(\mathcal{L}_{i})$ can be decomposed as the concatenation of a walk that first passes through a link in $\mathcal{L}_{i}$ and a walk that passes through $\mathcal{L}_{i}$ exactly $n$ times. Since every concatenation is unique and all unions are disjoint, taking the generating function on both sides yields:
\begin{equation*}
M[[\mathcal{\hat{W}}_{k}^{(n+1)}(\mathcal{L}_{i});x]]=x\cdot \sum_{l\in T}(m_{ki}\cdot M[[\mathcal{\hat{W}}_{l}^{(n)}(\mathcal{L}_{i});x]]+m_{kl}M[[\mathcal{\hat{W}}_{i}^{(n)}(\mathcal{L}_{i});x]])
\end{equation*}
for any $k$. Similarly:
\begin{equation*}
M[[\mathcal{\mathring{W}}_{k}^{(n+1)}(\mathcal{L}_{j});x]]=x\cdot \sum_{l\in T}(m_{kj}\cdot M[[\mathcal{\mathring{W}}_{l}^{(n)}(\mathcal{L}_{j});x]]+m_{kl}M[[\mathcal{\mathring{W}}_{j}^{(n)}(\mathcal{L}_{j});x]]).
\end{equation*}

By comparing these two equalities, we can establish an inductive proof: the dominance relationship holds for $n+1$ whenever it holds for $n$, given the conditions $m_{ii}\gtrsim m_{jj}$ and $m_{il}\gtrsim m_{jl}$ for all $l \notin \{i,j\}$. This approach, leveraging the generating functions of uniquely decomposed walk sets $\mathcal{\hat{W}}_{k}^{(n+1)}(\mathcal{L}_{i})$ and $\mathcal{\mathring{W}}_{k}^{(n+1)}(\mathcal{L}_{j})$, yields an elegant and concise proof of the proposition.

\section{Conclusion}

This paper develops a walk-counting framework based on generating functions to analyze how specific node and link sets influence walks between network nodes. At the core of our approach is the first passage decomposition, which establishes that walks from node $i$ to node $j$ passing through a given set can be uniquely decomposed into two fundamental components: walks from $i$ that first reach the set (without previously encountering it), concatenated with subsequent walks from the set to the target node $j$. This decomposition reveals a system of equations  that captures the relationships between walks under various passing or avoiding constraints in terms of their associated generating functions. By strategically solving this  system, with different parameters treated as known or unknown, we derive analytical results  with broad applicability.

We demonstrate the versatility of our framework through three primary applications. First, we analyze structural interventions—including node removal and link modifications—on walks between node pairs, deriving closed-form expressions that quantify intervention impacts throughout the network. Second, we extend \cite{Bramoulle2018}'s target centrality concept to scenarios with multiple senders and receivers. In addition, we provide precise measures of expected information transmission frequencies under various non-retransmission constraints. Finally, we employ first passage decomposition to prove a result about optimal link construction. When node $i$'s neighborhood contains node $j$'s neighborhood, establishing connections between any external node set and node $i$ generates more walks of any arbitrary length throughout the network than forming equivalent connections to node $j$. Throughout these applications, our approach transforms complex network analysis problems into tractable algebraic manipulations with clear walk-counting interpretations.

\newpage
\appendix
\section{Appendix}
\begin{proof}[Proof of Lemma \ref{pro:walk decomposition}] 

(i) First, we establish that $\mathcal{W}_{ij}(A) \supseteq \bigcup_{l\in A}\{\mathcal{W}_{il}(\lnot A)\backslash(i) \odot \mathcal{W}_{lj}\backslash(j)\}$. 

For any $l\in A$, consider walks $w_{il}=(i,i_1,...,i_{K-1},l) \in \mathcal{W}_{il}(\lnot A)\backslash(i)$ and $w_{lj}=(l,j_1,...,j_{K'-1},j) \in \mathcal{W}_{lj}\backslash(j)$. Their concatenation
\begin{equation*}
w_{il}\odot w_{lj}=(i,i_1,...,i_{K-1},l,j_1,...,j_{K'-1},j)
\end{equation*}
belongs to $\mathcal{W}_{ij}(A)$ since it passes through $l\in A$.

For the reverse inclusion, consider any walk $w=(i,i_1,...,i_{K-1},j) \in \mathcal{W}_{ij}(A)$. By definition, $w$ must pass through at least one node in $A$. Let $i_{k^*}$ be the first node in $w$ belonging to $A$, and denote this node as $l = i_{k^*}$. Note that $k^* \geq 1$ since $w \in \mathcal{W}_{ij}(A)$. The walk $w$ can be decomposed as:
\begin{equation*}
w=(i,i_1,...,i_{k^*-1},l) \odot (l,i_{k^*+1},...,i_{K-1},j)
\end{equation*}
where $(i,i_1,...,i_{k^*-1},l) \in \mathcal{W}_{il}(\lnot A)\backslash(i)$ since no node before $l$ belongs to $A$, and $(l,i_{k^*+1},...,i_{K-1},j)\in \mathcal{W}_{lj}\backslash(j)$.

To establish uniqueness of this decomposition, note that the concatenation contains no empty walks since the ending point of walks in $\mathcal{W}_{il}(\lnot A)\backslash(i)$ and the starting point of walks in $\mathcal{W}_{lj}\backslash(j)$ are identical (node $l$). We need to show that for any $l\in A$, if $w_{il}\odot w_{lj}=w_{il}^{\prime}\odot w_{lj}^{\prime}$ where $w_{il},w_{il}^{\prime}\in \mathcal{W}_{il}(\lnot A)\backslash(i)$ and $w_{lj},w_{lj}^{\prime}\in \mathcal{W}_{lj}\backslash(j)$, then $w_{il}=w_{il}^{\prime}$ and $w_{lj}=w_{lj}^{\prime}$.

Suppose, for contradiction, that $\#(w_{il}) > \#(w_{il}^{\prime})$. Then, $w_{il}=w_{il}^{\prime}\odot w_{ll}$ for some $w_{ll}\in \mathcal{W}_{ll}\backslash(l)$. This implies that $w_{il}$ contains $l$ as an intermediate node, which contradicts $w_{il}\in \mathcal{W}_{il}(\lnot A)\backslash(i)$ since $l \in A$. Moreover, if $\#(w_{il})=\#(w_{il}^{\prime})$, then $w_{il}=w_{il}^{\prime}$ and consequently $w_{lj}=w_{lj}^{\prime}$ since $w_{il}\odot w_{lj}=w_{il}^{\prime}\odot w_{lj}^{\prime}$.

(ii) For links, we first demonstrate that $\mathcal{W}_{ij}(\mathcal{L}) \subseteq \bigcup_{(l,k)\in \mathcal{L}}\{\mathcal{W}_{il}(\lnot \mathcal{L})\odot (l,k)\odot \mathcal{W}_{kj}\}$. 

Note that there are no length-zero walks in either set. For length-one walks, if $(i,j) \in \mathcal{W}_{ij}(\mathcal{L})$, then $(i,j) \in \mathcal{L}$, and
\begin{equation*}
(i,j) = (i) \odot (i,j) \odot (j) \in \mathcal{W}_{ii}(\lnot \mathcal{L}) \odot (i,j) \odot \mathcal{W}_{jj}
\end{equation*}

For walks of length $t\geq 2$, let $w=(i,i_1,...,i_{t-1},j)\in \mathcal{W}_{ij}(\mathcal{L})$. Since $w \in \mathcal{W}_{ij}(\mathcal{L})$, it must contain at least one link in $\mathcal{L}$. Let $(i_l,i_{l+1})$ be the first such link in $w$, where $l$ is the minimum index such that $(i_l,i_{l+1}) \in \mathcal{L}$. We use the convention that $i_0 = i$ and $i_t = j$. Then:

If $l=0$, $w = (i) \odot (i,i_1) \odot (i_1,...,i_{t-1},j) \in \mathcal{W}_{ii}(\lnot \mathcal{L}) \odot (i,i_1) \odot \mathcal{W}_{i_1j}$

If $l>0$, $w = (i,i_1,..,i_l) \odot (i_l,i_{l+1}) \odot(i_{l+1},...,j) \in \mathcal{W}_{ii_l}(\lnot \mathcal{L}) \odot (i_l,i_{l+1}) \odot \mathcal{W}_{i_{l+1}j}.$

For the reverse inclusion, length-one walks in $\bigcup_{(l,k)\in \mathcal{L}}\{\mathcal{W}_{il}(\lnot \mathcal{L}) \odot (l,k) \odot \mathcal{W}_{kj}\}$ must be in $\mathcal{L}$ and thus in $\mathcal{W}_{ij}(\mathcal{L})$. For walks of length $t\geq 2$:
\begin{equation*}
w = (i,i_1,...,i_l) \odot (i_l,i_{l+1}) \odot (i_{l+1},...,j) = (i,i_1,...,i_l,i_{l+1},...,j)
\end{equation*}
where $l\geq 0$ and $(i_l,i_{l+1}) \in \mathcal{L}$. Such a walk belongs to $\mathcal{W}_{ij}(\mathcal{L})$ by definition, since it contains the link $(i_l,i_{l+1}) \in \mathcal{L}$.

To establish uniqueness of this decomposition, note that any concatenation between a walk in $\mathcal{W}_{il}(\lnot \mathcal{L})$, the link $(l,k)$, and a walk in $\mathcal{W}_{kj}$ contains no empty components. Consider a walk $w = (i,i_1,...,j)$ and let $t^* = \min\{t: (i_t,i_{t+1}) \in \mathcal{L}\}$, using the convention $i_0 = i$. The index $t^*$ is unique by definition.

If $w = w_{il} \odot (l,k) \odot w_{kj} = w_{il}' \odot (l',k') \odot w_{k'j}'$ where $w_{il}, w_{il}' \in \mathcal{W}_{il}(\lnot \mathcal{L})$, $(l,k), (l',k') \in \mathcal{L}$, and $w_{kj}, w_{k'j}' \in \mathcal{W}_{kj}$, then $(i_{t^*}, i_{t^*+1}) = (l,k) = (l',k')$. Consequently, $(i,i_1,...,i_{t^*}) = w_{il} = w_{il}'$ and $(i_{t^*+1},...,j) = w_{kj} = w_{k'j}'$. The uniqueness of the decomposition follows from the uniqueness of $t^*$.
\end{proof}

\bigskip

\begin{proof}[Proof of Theorem \ref{pro:GF}]
(i) Since $\mathcal{W}_{ij}(\lnot A) \cap \mathcal{W}_{ij}(A) = \emptyset$ for any $i,j\in N$ and $A\subseteq N$, equation (\ref{eq:GF1}) follows from $\mathcal{W}_{ij} = \mathcal{W}_{ij}(\lnot A)\cup \mathcal{W}_{ij}(A)$ and the sum property of generating function.

By Lemma \ref{pro:walk decomposition},
\begin{equation*}
M[[\mathcal{W}_{ij}(A)]] = M[[\bigcup_{l\in A}\{\mathcal{W}_{il}(\lnot A)\backslash(i) \odot \mathcal{W}_{lj}\backslash(j)\}]].
\end{equation*}
For distinct nodes $l,l'\in A$, the sets $\{\mathcal{W}_{il}(\lnot A)\backslash(i)\odot \mathcal{W}_{lj}\backslash(j)\}$ and $\{\mathcal{W}_{il'}(\lnot A)\backslash(i)\odot \mathcal{W}_{l'j}\backslash(j)\}$ are disjoint since walks in $\mathcal{W}_{il}(\lnot A)$ and $\mathcal{W}_{il'}(\lnot A)$ cannot hit any node in $A$ except their end nodes. Therefore, by the properties of generating function,
\begin{equation*}
M[[\bigcup_{l\in A}\{\mathcal{W}_{il}(\lnot A)\backslash(i) \odot \mathcal{W}_{lj}\backslash(j)\}]] = \sum_{l\in A}M[[\{\mathcal{W}_{il}(\lnot A)\backslash(i) \odot \mathcal{W}_{lj}\backslash(j)\}]]
\end{equation*}

For any node $l\in A$, the decomposition $\mathcal{W'}=\mathcal{W}_{il}(\lnot A)\backslash(i)\odot \mathcal{W}_{lj}\backslash(j)$ is unique. Thus, by the properties of generating function,
\begin{equation*}
M[[\{\mathcal{W}_{il}(\lnot A)\backslash(i) \odot \mathcal{W}_{lj}\backslash(j)\}]] = M[[\mathcal{W}_{il}(\lnot A)\backslash(i)]]\cdot M[[\mathcal{W}_{lj}\backslash(j)]]
\end{equation*}

Note that the zero-length walk $(i)\in \mathcal{W}_{il}(\lnot A)$ when $l=i$ and $(i)\in \mathcal{W}_{lj}$ when $l=j$. Hence,
\begin{align*}
M[[\mathcal{W}_{il}(\lnot A)\backslash(i)]] &= M[[\mathcal{W}_{il}(\lnot A)]] - \mathbf{1}_{l=i}\\
M[[\mathcal{W}_{lj}\backslash(j)]] &= M[[\mathcal{W}_{lj}]] - \mathbf{1}_{l=j}
\end{align*}
Equation (\ref{eq:GF2}) follows from combining these equalities.

(ii) Equation (\ref{eq:GF3}) follows directly from applying generating function to equation (\ref{eq:partition2}).

By Lemma \ref{pro:walk decomposition},
\begin{align*}
M[[\mathcal{W}_{ij}(\mathcal{L})]] &= M[[\bigcup_{(l,k)\in \mathcal{L}}\{\mathcal{W}_{il}(\lnot \mathcal{L})\odot(l,k)\odot \mathcal{W}_{kj}\}]]\\
&= \sum_{(l,k)\in \mathcal{L}}M[[\mathcal{W}_{il}(\lnot \mathcal{L})\odot(l,k)\odot \mathcal{W}_{kj}]].
\end{align*}
The second equality holds since $\{\mathcal{W}_{il}(\lnot \mathcal{L})\odot(l,k)\} \cap \{\mathcal{W}_{il'}(\lnot \mathcal{L})\odot(l',k')\} = \emptyset$ for distinct $(l,k),(l',k')\in \mathcal{L}$.

For any walks $w_{il}\in \mathcal{W}_{il}(\lnot \mathcal{L})$ and $w_{kj}\in \mathcal{W}_{kj}$, $w_{il}\odot(l,k)\odot w_{kj}\neq 0$. Therefore,
\begin{align*}
\sum_{(l,k)\in \mathcal{L}}M[[\mathcal{W}_{il}(\lnot \mathcal{L})\odot(l,k)\odot \mathcal{W}_{kj}]] 
&= \sum_{(l,k)\in \mathcal{L}}M[[\mathcal{W}_{il}(\lnot \mathcal{L})]]\cdot M[[(l,k)]]\cdot M[[\mathcal{W}_{kj}]]\\
&= \sum_{(l,k)\in \mathcal{L}}M[[\mathcal{W}_{il}(\lnot \mathcal{L})]]\cdot M[[\mathcal{W}_{kj}]]\cdot x
\end{align*}
where the last equality uses the fact that $(l,k)$ has length 1.
\end{proof}

\bigskip
\begin{proof}[Proof of Proposition \ref{cor:nodes}]
Suppose $i \in A$ and $j \notin A$. Rewriting equation \eqref{eq:GF2} in matrix form, we obtain
\begin{equation*}
\mathbf{M}_{AA^C}^A = (\mathbf{M}_{AA}^{\lnot A} - \mathbf{I})\mathbf{M}_{AA^C} = (\mathbf{I} - (\mathbf{M}_{AA})^{-1})\mathbf{M}_{AA^C}
\end{equation*}
Since $\mathbf{M}_{AA^C}^A = \mathbf{M}_{AA^C} - \mathbf{M}_{AA^C}^{\lnot A}$, we have
\begin{align*}
\mathbf{M}_{AA^C} - \mathbf{M}_{AA^C}^{\lnot A} &= (\mathbf{I} - (\mathbf{M}_{AA})^{-1})\mathbf{M}_{AA^C} \\
\Rightarrow \mathbf{M}_{AA^C}^{\lnot A} &= (\mathbf{M}_{AA})^{-1}\mathbf{M}_{AA^C}
\end{align*}
When $i \notin A$ and $j \in A$, a similar argument yields $\mathbf{M}_{A^CA}^{\lnot A} = \mathbf{M}_{A^CA}(\mathbf{M}_{AA})^{-1}$.

For $i \notin A$ and $j \notin A$, the matrix form of equation \eqref{eq:GF2} is 
\begin{equation*}
\mathbf{M}_{A^CA^C}^A = \mathbf{M}_{A^CA}^{\lnot A}\mathbf{M}_{AA^C}
\end{equation*}
Substituting $\mathbf{M}_{A^CA^C}^A = \mathbf{M}_{A^CA^C} - \mathbf{M}_{A^CA^C}^{\lnot A}$ and our result on $\mathbf{M}_{A^CA}^{\lnot A}$, we get
\begin{align*}
\mathbf{M}_{A^CA^C} - \mathbf{M}_{A^CA^C}^{\lnot A} &= \mathbf{M}_{A^CA}^{\lnot A}\mathbf{M}_{AA^C} \\
\Rightarrow \mathbf{M}_{A^CA^C}^{\lnot A} &= \mathbf{M}_{A^CA^C} - \mathbf{M}_{A^CA}^{\lnot A}\mathbf{M}_{AA^C} \\
\Rightarrow \mathbf{M}_{A^CA^C}^{\lnot A} &= \mathbf{M}_{A^CA^C} - \mathbf{M}_{A^CA}(\mathbf{M}_{AA})^{-1}\mathbf{M}_{AA^C}
\end{align*}

To show equation (\ref{eq:nodes2}), note that 
\begin{eqnarray*}
&&\sum_{i,j\in N}m_{ij}-\sum_{i,j\in A^{C}}m_{ij}^{\lnot A} =\mathbf{1}%
^{\prime }\mathbf{M}_{A^{C}A}(\mathbf{M}_{AA})^{-1}\mathbf{M}_{AA^{C}}%
\mathbf{1}+\mathbf{1}^{\prime }\mathbf{M}_{A^{C}A}\mathbf{1}+\mathbf{1}%
^{\prime }\mathbf{M}_{AA}\mathbf{1}+\mathbf{1}^{\prime }\mathbf{M}_{AA^{C}}%
\mathbf{1} \\
&=&\mathbf{1}^{\prime }\mathbf{M}_{A^{C}A}(\mathbf{M}_{AA})^{-1}\mathbf{M}%
_{AA^{C}}\mathbf{1+1}^{\prime }\mathbf{M}_{A^{C}A}(\mathbf{M}_{AA})^{-1}%
\mathbf{M}_{AA}\mathbf{1}\\&&+\mathbf{1}^{\prime }\mathbf{M}_{AA}(\mathbf{M}_{AA})^{-1}%
\mathbf{M}_{AA}\mathbf{1+1}^{\prime }\mathbf{M}_{AA}(\mathbf{M}_{AA})^{-1}%
\mathbf{M}_{AA^{C}}\mathbf{1} \\
&=&\mathbf{1}^{\prime }\mathbf{M}_{A^{C}A}(\mathbf{M}_{AA})^{-1}\left( 
\mathbf{M}_{AA^{C}}\mathbf{1+M}_{AA}\mathbf{1}\right) +\mathbf{1}^{\prime }%
\mathbf{M}_{AA}(\mathbf{M}_{AA})^{-1}\left( \mathbf{M}_{AA}\mathbf{1+M}%
_{AA^{C}}\mathbf{1}\right)  \\
&=&\mathbf{1}^{\prime }\mathbf{M}_{A^{C}A}(\mathbf{M}_{AA})^{-1}\mathbf{\dot{%
b}}_{A}+\mathbf{1}^{\prime }\mathbf{M}_{AA}(\mathbf{M}_{AA})^{-1}\mathbf{%
\dot{b}}_{A} \\
&=&\left( \mathbf{1}^{\prime }\mathbf{M}_{A^{C}A}+\mathbf{1}^{\prime }%
\mathbf{M}_{AA}\right) (\mathbf{M}_{AA})^{-1}\mathbf{\dot{b}}_{A} \\
&=&\mathbf{\mathring{b}}_{A}^{\prime }(\mathbf{M}_{AA})^{-1}\mathbf{\dot{b}}%
_{A}\text{.}
\end{eqnarray*}
\end{proof}

\bigskip
\begin{proof}[Proof of Corollary \ref{cor:link}]
We begin by expressing equation \eqref{eq:GF4} in matrix form. For a set of links $\mathcal{L}$, we have
\begin{equation*}
\mathbf{M}^{\mathcal{L}} = x\mathbf{M}^{\lnot \mathcal{L}}\mathbf{L}\mathbf{M}
\end{equation*}
where $\mathbf{L}$ is the adjacency matrix of the subgraph defined by links in $\mathcal{L}$, with elements $L_{ij} = 1$ if $(i,j) \in \mathcal{L}$ and $0$ otherwise.

Since $\mathbf{M}^{\mathcal{L}} = \mathbf{M} - \mathbf{M}^{\lnot \mathcal{L}}$, we can get
$\mathbf{M}^{\lnot \mathcal{L}} = \mathbf{M}(\mathbf{I} + x\mathbf{LM})^{-1}$ and $
\mathbf{M}^{\mathcal{L}} = x\mathbf{M}(\mathbf{I} + x\mathbf{LM})^{-1}\mathbf{LM}$.

Therefore, for the walks that avoid link set $\mathcal{L}_-$,
\begin{equation*}
\mathbf{M}(\mathcal{L}_-, \emptyset) = \mathbf{M}(\mathbf{I} + x\mathbf{L}_-\mathbf{M})^{-1}
\end{equation*}

Next, we determine the generating function matrix for walks that avoid $\mathcal{L}_-$ but must pass through at least one link in $\mathcal{L}_+$:
\begin{equation*}
\mathbf{M}(\mathcal{L}_-, \mathcal{L}_+) = (\mathbf{I} - x\mathbf{M}(\mathcal{L}_-, \emptyset)\mathbf{L}_+)^{-1}\mathbf{M}(\mathcal{L}_-, \emptyset)
\end{equation*}

Substituting the expression of $\mathbf{M}(\mathcal{L}_-, \emptyset)$ in, we get
\begin{eqnarray*}
\mathbf{M}\left( \mathcal{L}_{-},\mathcal{L}_{+}\right)  &=&\left( \mathbf{I}%
-x\mathbf{M}\left( \mathbf{I}+x\mathbf{L}_{-}\mathbf{M}\right) ^{-1}\mathbf{L%
}_{+}\right) ^{-1}\mathbf{M}\left( \mathbf{I}+x\mathbf{L}_{-}\mathbf{M}%
\right) ^{-1} \\
&=&\left( \mathbf{M}^{-1}-x\left( \mathbf{L}_{+}-\mathbf{L}_{-}\right)
\right) ^{-1} \\
&=&\mathbf{M}+x\mathbf{M}\left( \mathbf{L}_{+}-\mathbf{L}_{-}\right) \left( 
\mathbf{I}-x\left( \mathbf{L}_{+}-\mathbf{L}_{-}\right) \mathbf{M}\right)
^{-1}\mathbf{M}\text{.}
\end{eqnarray*}
\end{proof}

\bigskip
\begin{proof}[Proof of Proposition \ref{pro:target}]
The expressions for expected message reception counts follow directly from their definitions:
\begin{align*}
\xright{n}_{iA} &= \sum_{j\in A}\sum_{t=0}^{\infty}\delta^t|\mathcal{W}_{ij}^t(\lnot\{i\})| = \sum_{j\in A}m_{ij}^{\lnot\{i\}}\\
\xleft{n}_{iA} &= \sum_{j\in A}\sum_{t=0}^{\infty}\delta^t|\mathcal{W}_{ij}^t(\lnot A)| = \sum_{j\in A}m_{ji}^{\lnot A}\\
\xboth{n}_{iA} &= \sum_{j\in A}\sum_{t=0}^{\infty}\delta^t|\mathcal{W}_{ij}^t(\lnot(A\cup\{i\}))| = \sum_{j\in A}m_{ij}^{\lnot(A\cup\{i\})}
\end{align*}

For walks within set $A$, we rewrite equation \eqref{eq:GF2} in matrix form:
\begin{equation*}
\mathbf{M}_{AA}^{A}=\left( \mathbf{M}_{AA}^{\lnot A}-\mathbf{I}\right)\left( \mathbf{M}_{AA}-\mathbf{I}\right)
\end{equation*}

From equation \eqref{eq:GF1}, we know that $\mathbf{M}_{AA}=\mathbf{M}_{AA}^{A}+\mathbf{M}_{AA}^{\lnot A}$. Combining these equations:
\begin{align*}
\mathbf{M}_{AA}-\mathbf{M}_{AA}^{\lnot A} &= \left( \mathbf{M}_{AA}^{\lnot A}-\mathbf{I}\right) \left( \mathbf{M}_{AA}-\mathbf{I}\right)\\
\Rightarrow \mathbf{M}_{AA}^{\lnot A} &= 2\mathbf{I}-\left( \mathbf{M}_{AA}\right)^{-1}
\end{align*}

From Proposition \ref{cor:nodes}, we have established that $\mathbf{M}_{AA^C}^{\lnot A} = (\mathbf{M}_{AA})^{-1}\mathbf{M}_{AA^C}$ and $\mathbf{M}_{A^CA}^{\lnot A} = \mathbf{M}_{A^CA}(\mathbf{M}_{AA})^{-1}$.

For the single-node case where $A = \{j\}$, the results follow from applying the expressions of $\mathbf{M}_{AA^C}^{\lnot A}$, $\mathbf{M}_{A^CA}^{\lnot A}$ and $\mathbf{M}_{AA}^{\lnot A}$ with appropriate substitutions of $A = \{i\}$, $\{j\}$, and $\{i,j\}$ respectively.
\end{proof}

\bigskip

\begin{proof}[Proof of Corollary \ref{corr:component-intermediary}] 
The first equality follows from the fact that 
\begin{equation*}
m_{ij}^{\lnot\{i,j,k\}} = \left(2\mathbf{I}-(\mathbf{M}^{\{i,j,k\}})^{-1}\right)_{ij}.
\end{equation*}
The second equality can be verified by substituting $m_{ij}^{\{k\}}$, $m_{ii}^{\{k\}}$, $m_{jj}^{\{k\}}$ according to $\mathbf{M}_{A^CA^C}^A = \mathbf{M}_{A^CA}(\mathbf{M}_{AA})^{-1}\mathbf{M}_{AA^C}$.
\end{proof}

\bigskip

\begin{proof}[Proof of Proposition \ref{pro:belhaj}]
In the main text, we have shown that
\begin{equation*}
M\left[\left[ \mathcal{\hat{W}}_{k}^{(n+1)}(\mathcal{L}_{i});x\right]\right] = x \cdot \sum_{l\in T}\left(m_{ki} \cdot M\left[\left[\mathcal{\hat{W}}_{l}^{(n)}(\mathcal{L}_{i});x\right]\right] + m_{kl}M\left[\left[\mathcal{\hat{W}}_{i}^{(n)}(\mathcal{L}_{i});x\right]\right]\right)
\end{equation*}
\begin{equation*}
M\left[\left[\mathcal{\mathring{W}}_{k}^{(n+1)}(\mathcal{L}_{j});x\right]\right] = x \cdot \sum_{l\in T}\left(m_{kj} \cdot M\left[\left[\mathcal{\mathring{W}}_{l}^{(n)}(\mathcal{L}_{j});x\right]\right] + m_{kl}M\left[\left[\mathcal{\mathring{W}}_{j}^{(n)}(\mathcal{L}_{j});x\right]\right]\right),
\end{equation*}
for any node $k$ and any integer $n\geq 0$. 

Therefore, for any $k\notin \{i,j\}$, we have
\begin{equation*}
M\left[\left[\mathcal{\hat{W}}_{k}^{(n+1)}(\mathcal{L}_{i});x\right]\right] \gtrsim M\left[\left[\mathcal{\mathring{W}}_{k}^{(n+1)}(\mathcal{L}_{j});x\right]\right]
\end{equation*}
whenever $m_{ki}\gtrsim m_{kj}$ and the following inductive conditions hold:
\begin{equation*}
M\left[\left[\mathcal{\hat{W}}_{l}^{(n)}(\mathcal{L}_{i});x\right]\right] \gtrsim M\left[\left[\mathcal{\mathring{W}}_{l}^{(n)}(\mathcal{L}_{j});x\right]\right]
\end{equation*}
and
\begin{equation*}
M\left[\left[\mathcal{\hat{W}}_{i}^{(n)}(\mathcal{L}_{i});x\right]\right] \gtrsim M\left[\left[\mathcal{\mathring{W}}_{j}^{(n)}(\mathcal{L}_{j});x\right]\right].
\end{equation*}

Moreover, we have
\begin{equation*}
M\left[\left[\mathcal{\hat{W}}_{i}^{(n+1)}(\mathcal{L}_{i});x\right]\right] + M\left[\left[\mathcal{\hat{W}}_{j}^{(n+1)}(\mathcal{L}_{i});x\right]\right] \geq M\left[\left[\mathcal{\mathring{W}}_{i}^{(n+1)}(\mathcal{L}_{j});x\right]\right] + M\left[\left[\mathcal{\mathring{W}}_{j}^{(n+1)}(\mathcal{L}_{j});x\right]\right]
\end{equation*}
whenever $m_{ii}\gtrsim m_{ij}$ and
\begin{equation*}
M\left[\left[\mathcal{\hat{W}}_{i}^{(n)}(\mathcal{L}_{i});x\right]\right] + M\left[\left[\mathcal{\hat{W}}_{j}^{(n)}(\mathcal{L}_{i});x\right]\right] \geq M\left[\left[\mathcal{\mathring{W}}_{i}^{(n)}(\mathcal{L}_{j});x\right]\right] + M\left[\left[\mathcal{\mathring{W}}_{j}^{(n)}(\mathcal{L}_{j});x\right]\right].
\end{equation*}
\end{proof}

 \newpage

\bibliographystyle{chicagoa}
\bibliography{gen_fn}

\end{document}